\PassOptionsToPackage{x11names}{xcolor}
\documentclass[a4paper,cleveref, autoref, 11pt, fleqn, thm-restate]{article}
\RequirePackage{fullpage}
\usepackage{subcaption}
\usepackage{makecell}
\usepackage{amsthm,amsmath,amssymb}
\usepackage{caption}
\usepackage{complexity}
\usepackage{enumerate}
\usepackage{enumitem}
\usepackage{microtype}

\usepackage{booktabs}
\usepackage{hyperref}
\usepackage{cleveref}
\usepackage[T1]{fontenc}

\usepackage[
final,  
todonotes={disable,colorinlistoftodos,prependcaption,textsize=footnotesize}]{changes} 

\newcommand{\clu}[2][]{\todo[author=Lu, color=HotPink3!40,bordercolor=HotPink3, #1]{#2}}

\definecolor{pyC0}{HTML}{1f77b4}
\definecolor{pyC1}{HTML}{ff7f0e}
\definecolor{pyC2}{HTML}{2ca02c}
\usepackage{tikz}
\usepackage{physics}
\usepackage[american,siunitx]{circuitikz}

\usepackage[compat=0.6]{yquant}
\usepackage{import}
\usetikzlibrary{angles, quotes,shapes}
\usepackage{hyperref}
\usetikzlibrary{backgrounds}
\useyquantlanguage{qasm}

\usepackage[linesnumbered, noend,ruled,vlined]{algorithm2e}
\usetikzlibrary{backgrounds}
\usetikzlibrary{fit, quotes}
\usetikzlibrary{backgrounds}

\usetikzlibrary{fit, quotes}

\newcommand{\ptt}{perpane}
\newcommand{\Call}[2]{#1($#2$)}

\usepackage[affil-it]{authblk}

\theoremstyle{plain} 
\newtheorem{theorem}{Theorem}[section]
\newtheorem{lemma}[theorem]{Lemma}

\newtheorem{observation}[theorem]{Observation}
\newtheorem{proposition}[theorem]{Proposition}

\newtheorem*{definition}{Definition}
\newtheorem{conjecture}{Conjecture}

\newcommand\fixed[1]{}

\newcommand\myccnot[3]{
\yquant
["north east:\textcolor{#1}{$\{#2\}$}" {font=\protect\scriptsize, inner sep=0mm}, style={#1}]#3
}

\newcommand\mycnot[2]{
\yquant
["north east:$\{#1\}$" {font=\protect\scriptsize, inner sep=0mm}]#2
}
\usepackage{enumitem}

\let\oldket\ket
\renewcommand{\ket}[1]{\{#1\}}
\newcommand{\rowadd}[1]{\textsc{cnot}(#1)}

\tikzstyle{filled vertex}  = [{circle,draw=black,fill=purple!80,inner sep=1pt}]
\tikzstyle{empty vertex}  = [{circle, draw, fill = white, inner sep=1pt, minimum width=1.5pt}]
\usepackage{listofitems}
\definecolor{purplecolor}{rgb}{0.65,0.12,0.82} 
\usepackage[T1]{fontenc}
\usepackage[scale=0.9]{merriweather}
\usepackage{charter,eulervm}

\title{Toward Minimum Graphic Parity Networks}
\author[2]{Yixin Cao}
\author[1,3]{Yiren Lu}
\author[1,3]{Junhong Nie}
\author[1,3]{Xiaoming Sun}
\author[1,3]{Guojing Tian}

\affil[1]{State Key Lab of Processors, Institute of Computing Technology, Chinese Academy of Sciences, Beijing, China}
\affil[2]{Department of Computing, Hong Kong Polytechnic University, Hong Kong, China}
\affil[3]{School of Computer Science and Technology, University of Chinese Academy of Sciences, Beijing, China}

\crefname{algocf}{alg.}{algs.}
\Crefname{algocf}{Algorithm}{Algorithms}
\date{}

\begin{document}
\maketitle

\listoftodos

\begin{abstract}
  Quantum circuits composed of CNOT and $R_z$ are fundamental building blocks of many quantum algorithms, so optimizing the synthesis of such quantum circuits is crucial. We address this problem from a theoretical perspective by studying the graphic parity network synthesis problem. A graphic parity network for a graph~$G$ is a quantum circuit composed solely of CNOT gates %
  where each edge of~$G$ is represented in the circuit, and the final state of the wires matches the original input. We aim to synthesize graphic parity networks with the minimum number of gates, specifically for quantum algorithms addressing combinatorial optimization problems with Ising formulations.  We demonstrate that a graphic parity network for a connected graph with~$n$ vertices and~$m$ edges requires at least~$m+n-1$ gates. This lower bound can be improved to $m+\Omega(m) = m+\Omega(n^{1.5})$ when the shortest cycle in the graph has a length of at least five.  We complement this result with a simple randomized algorithm that synthesizes a graphic parity network with expected~$m + O(n^{1.5}\sqrt{\log n})$ gates.

  Additionally, we begin exploring connected graphs that allow for graphic parity networks with exactly~$m+n-1$ gates. We conjecture that all such graphs belong to a newly defined graph class.  Furthermore, we present a linear-time algorithm for synthesizing minimum graphic parity networks for graphs within this class.  However, this graph class is not closed under taking induced subgraphs, and we show that recognizing it is \NP-complete, which is complemented with a fixed-parameter tractable algorithm parameterized by the treewidth.
\end{abstract}

\section{Introduction}\label{sec:intro}

Over the past decade, there has been significant progress in quantum computation, with advancements in both experimental and theoretical aspects~\cite{bravyiHighthresholdLowoverheadFaulttolerant2024,kallaugherExponentialQuantumSpace2023a,niBeatingBreakevenPoint2023,caoGenerationGenuineEntanglement2023}.  Compared to classical circuits, quantum circuits are more sensitive to noise, which is one of the bottlenecks limiting the scalability of quantum computers. Each quantum gate in a circuit may introduce some noise to the output quantum state, and as the circuit size increases, the accumulated noise can overwhelm any meaningful computational results.  Since we are still in the noisy intermediate-scale quantum (\textsc{nisq}) era, there exists a fundamental necessity to minimize the number of gates used in quantum circuits~\cite{Mottonen2004a,Vartiainen2004,10.5555/3381089.3381102,Wu2023}, especially for two-qubit gates, which are even more vulnerable to noise than single-qubit gates~\cite{caoGenerationGenuineEntanglement2023,Li2023}.

{%
Subcircuits consisting only of $\textsc{cnot}$ and $R_z$ gates appear in many quantum algorithms, such as quantum simulation \cite{namAutomatedOptimizationLarge2018}, quantum approximate optimization algorithm (\textsc{qaoa}) \cite{farhiQuantumApproximateOptimization2014a}, variational quantum eigensolver (\textsc{vqe})~\cite{peruzzo2014VQE}, etc. To optimize such quantum circuits with massive $\{\textsc{cnot}, R_z\}$-only components, a widely employed method is to independently resynthesize these components, and this problem has been extensively treated by \textit{parity network synthesis} \cite{amyControlledNOTComplexityControlledNOT2018,gheorghiuReducingCNOTCount2023}.
Parity network is a \textsc{cnot}-only quantum circuit. In a sense, \textsc{cnot} gates can be viewed as the quantum analog of classical \textsc{xor} gates.\footnote{To make our work accessible to a general audience, we will minimize the use of quantum computing jargon in the main text.  Consequently, our definitions and explanations might seem imprecise to quantum experts.  For a more detailed and rigorous treatment, please refer to \cref{app:quantum-intro}.}
A \textsc{cnot} gate applied to qubit wires~$i$ and~$j$, denoted by~$\rowadd{i,j}$, changes the value of wire~$j$ to~$a\oplus b$, where~$a,b\in \mathbb{F}_2$ are the values of wires~$i$ and~$j$, respectively. Here, wire~$i$ is the \textit{control wire}, and wire~$j$ is the \textit{target wire}.
A \textsc{cnot} circuit is called a \emph{parity network} for a set family~$S$ if every element of~$S$ appears in the circuit as a \textit{term} at some intermediate point, and the final values of the wires are the same as the original input, denoted by ${x_1,x_2,\dots,x_n \in \mathbb{F}_2}$~\cite{amyControlledNOTComplexityControlledNOT2018}, and we say that a wire has the \textit{term}~$A = \ket{a_{1}, \ldots, a_{|A|}}$ on it if the value of a wire evaluates to~$x_{a_{1}}\oplus \cdots \oplus x_{a_{|A|}}$.

The set family~$S$ defines a hypergraph.  In certain applications, it is actually a graph, namely all terms in $S$ are two-element sets. 
For example, the maximum cut problem asks for a bipartition of the vertex set of a graph such that the number of edges between the two parts is maximized.
The well-known quantum adiabatic algorithm (\textsc{qaa}) \cite{farhiQuantumComputationAdiabatic2000a} and \textsc{qaoa} \cite{farhiQuantumApproximateOptimization2014a} use the following formulation to solve the maximum cut problem on a graph~$G$:
\[
  \max_{x\in \{-1,1\}^{|V(G)|}}\sum_{(u,v)\in E(G)} \frac{1-x_ux_v}{2}.
\]
A principal component of the quantum circuit of these algorithms is a parity network with each term being a two-element set. A similar approach has been taken by many authors~\cite{wangQuantumApproximateOptimization2018, wurtzMaxCutQuantumApproximate2021, wurtzFixedangleConjecturesQuantum2021, zhouQuantumApproximateOptimization2020, bassoPerformanceLimitationsQAOA2022a, bassoQuantumApproximateOptimization2022a, anshuConcentrationBoundsQuantum2023c,shaydulinEvidenceScalingAdvantage2024}, and there have been preliminary physical demonstrations of these quantum algorithms on \textsc{nisq} devices \cite{paganoQuantumApproximateOptimization2020, harriganQuantumApproximateOptimization2021a}.
Recently, heuristic attempts have been made to synthesize small parity networks~\cite{amyControlledNOTComplexityControlledNOT2018,gheorghiuReducingCNOTCount2023, debrugiereFasterShorterSynthesis2024}.  As far as we know, however, none of them is accompanied with analysis of performance.
Such parity networks where all terms in $S$ are two-element sets are called \emph{graphic parity networks}, and see \cref{fig:pn-example} as examples. This paper is mainly concerned with the graphic parity network synthesis problem.

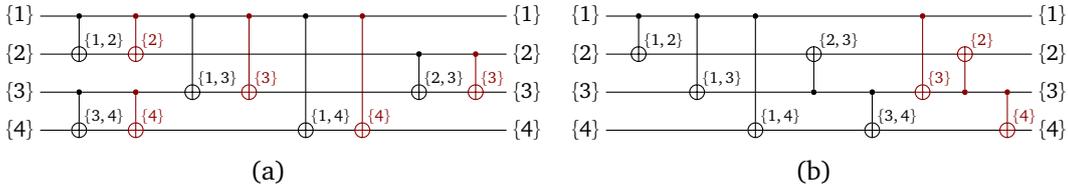
\begin{figure}[ht!]
  \centering

  \subcaptionbox{\label{fig:trivial-gpn}}{
    \centering
    \begin{tikzpicture}[scale=.75]
      \begin{yquant}[operator/separation=2.3mm]
        qubit {$\{1\}$} x[1];
        qubit {$\{2\}$} x[+1];
        qubit {$\{3\}$} x[+1];
        qubit {$\{4\}$} x[+1];

        \mycnot{\protect{1, 2}}{cnot x[1] | x[0]};

        \myccnot{Red4}{2}{cnot x[1] | x[0]};

        \mycnot{\protect{3, 4}}{cnot x[3] | x[2]};

        \myccnot{Red4}{4}{cnot x[3] | x[2]};

        \mycnot{\protect{1, 3}}{cnot x[2] | x[0]};
        \myccnot{Red4}{3}{cnot x[2] | x[0]};

        \mycnot{\protect{1, 4}}{cnot x[3] | x[0]};
        \myccnot{Red4}{4}{cnot x[3] | x[0]};

        \mycnot{\protect{2, 3}}{cnot x[2] | x[1]};
        \myccnot{Red4}{3}{cnot x[2] | x[1]};

        output {$\{1\}$} x[0];
        output {$\{2\}$} x[1];
        output {$\{3\}$} x[2];
        output {$\{4\}$} x[3];
      \end{yquant}
    \end{tikzpicture}
  }
  \subcaptionbox{\label{fig:optimal-gpn}}{
    \centering
    \begin{tikzpicture}[scale=.75]
      \begin{yquant}[operator/separation=1.2mm]
        qubit {$\{1\}$} x[1];
        qubit {$\{2\}$} x[+1];
        qubit {$\{3\}$} x[+1];
        qubit {$\{4\}$} x[+1];

        \mycnot{\protect{1, 2}}{cnot x[1] | x[0]};

        \mycnot{\protect{1, 3}}{cnot x[2] | x[0]};

        \mycnot{\protect{1, 4}}{cnot x[3] | x[0]};

        \mycnot{\protect{2, 3}}{cnot x[1] | x[2]};

        \mycnot{\protect{3, 4}}{cnot x[3] | x[2]};

        \myccnot{Red4}{3}{cnot x[2] | x[0]};

        \myccnot{Red4}{2}{cnot x[1] | x[2]};
        \myccnot{Red4}{4}{cnot x[3] | x[2]};

        output {$\{1\}$} x[0];
        output {$\{2\}$} x[1];
        output {$\{3\}$} x[2];
        output {$\{4\}$} x[3];
      \end{yquant}
    \end{tikzpicture}
  }
  \caption{\added{%
      Two graphic parity networks for the graph with edges $\{\{1, 2\}, \{1, 3\}, \{1, 4\}, \{2, 3\}, \{3, 4\}\}$, shown in \cref{fig:diamond}.} {Each line represents a qubit wire labeled at the left.}
    We use~$i \boldsymbol{\cdot}\!\!-\!\!\oplus j$ to denote~$\rowadd{i,j}$.
    Indicated at the above right corner of~$\oplus$ is the resulting term of this operation.  {The black terms correspond to the sets, and the red do not.} %
  }
  \label{fig:pn-example}
\end{figure}
}

\begin{figure}[ht!]
  \centering\small
  \subcaptionbox{\label{fig:diamond}}{
    \centering
    \begin{tikzpicture}[every node/.style={empty vertex}]
      \node (v1) at (3, 1) {$1$};
      \foreach \i in {2, 3, 4} {
          \node (v\i) at (\i, 0) {$\i$};
          \draw (v1) -- (v\i);
        }
      \draw (v2) -- (v3) -- (v4);
    \end{tikzpicture}
  }
  \hspace{1cm}
  \subcaptionbox{\label{fig:k2p}}{
    \centering
    \begin{tikzpicture}
      \foreach \i in {1, 2}
      \node[empty vertex] (v\i) at ({\i+1.5}, 1) {$\i$};
      \foreach[count=\i from 3] \x/\l in {1/3, 2/4, 3/5, 5/p} {
          \node[empty vertex] (u\i) at (\x, 0) {$\l$};
          \draw (v1) -- (u\i) -- (v2);
        }
      \node at (4, 0) {$\cdots$};
    \end{tikzpicture}
  }
  \caption{(a) A chordal graph, and (b)~$K_{2, p}$, which can be made chordal by adding the edge~$(1, 2)$.}
  \label{fig:example-graphs}
\end{figure}
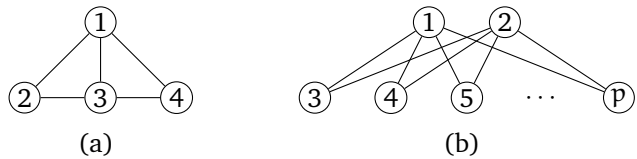

\paragraph*{Our contributions.}

The number of gates in a graphic parity network is called its \emph{size}.
Let~$G$ be the input graph, and let~$n$ and~$m$ denote the numbers of vertices and edges, respectively, in~$G$.
A trivial graphic parity network is in the fashion of~\Cref{fig:trivial-gpn}; i.e., generating an edge and cleaning it up immediately.  This leads to a trivial upper bound of~$2 m$, which is tight, as witnessed by graphs in which each component consists of one or two vertices.
On the other hand, the definition implies a trivial lower bound,~$m$, which, although tight, is uninteresting because it can only hold on edgeless graphs.

Our first result is a nontrivial lower bound of the size of graphic parity networks.

\begin{theorem}
  \label{thm:opt-pn-size-ctrb}
  A graphic parity network for a graph~$G$ contains at least~$m+n-c$ gates, where~$c$ is the number of components of~$G$.
\end{theorem}

Again, this bound is tight: it matches the upper bound~$2 m$ for all forests.
We say that a graphic parity network is \emph{perfect} if its size is precisely~$m+n-c$.
A graph admits a graphic parity network that is perfect or close to perfect only when there are a lot of operations from two edge terms to another edge term, e.g.,~$\ket{1, 3} \oplus\ket{1, 2} = \ket{2, 3}$ in \Cref{fig:optimal-gpn}.  The element~$1$ is canceled from the target wire of this operation, which is hence called a \emph{cancellation}.  A cancellation can only happen when the three edges form a triangle.
As we will see, for \emph{chordal graphs}---graphs in which every induced cycle is a triangle,---it is easy to synthesize perfect graphic parity networks.  This observation can be extended to graphs that can be turned into chordal graphs by adding a small number of edges, which contains many cycles of length three or four.  See \Cref{fig:example-graphs} for examples.  On the other hand, if the length of shortest cycles in a graph is five or more, we cannot do significantly better than the trivial construction of~$2 m$ gates.

\begin{theorem} %
  \label{thm:hel_nsqrt}
  For any positive integer~$n$, there exists a graph~$G$ with~$\Omega\left(n\sqrt n\right)$ edges such that the size of its minimum graphic parity networks is~$m + \Omega(m)$.
\end{theorem}

Motivated by these observations, we propose a randomized algorithm for synthesizing graphic parity networks.
It processes the vertices in a random order, and for each vertex~$v$, generates all the edges between~$v$ and latter vertices.  The algorithm tries to minimize the size by exploiting the triangles and squares in the input graph.
The expected size of the synthesized \replaced{circuit}{curcuit} almost matches the bound in \Cref{thm:hel_nsqrt}.  In particular, it produces very good results for dense graphs.

\begin{theorem}
  \label{cor:simpl}
  There is a polynomial-time algorithm synthesizing a graphic parity network with expected size $m + O\left(n^{1.5}\sqrt{\log n}\right)$.
  Moreover, if all but a constant number of vertices have degrees~$\Omega(n)$, the synthesized graphic parity network has expected size~$m + O\left(n \log n\right)$.
\end{theorem}
\clu[inline]{\textit{Reviewer C: The statement of theorem 1.2 goes well beyond the point where implicitly using m and n to refer to graph parameters makes sense, given that both of those are separately bound in the theorem statement.} What does he mean?}

A natural question is to characterize graphs admitting PERfect graphic PArity NEtworks, which we call \emph{\ptt{}}.
An immediate consequence of Theorem~\ref{thm:opt-pn-size-ctrb} is that in a perfect graphic parity network, there always exists a wire that is not a target of any operation.  We noticed that for all \ptt{} graphs we have discovered, we can synthesize a perfect graphic parity network in which there exists a wire that is not a \replaced{control }{source } of any operation, e.g., wires~$2$ and~$4$ in \Cref{fig:optimal-gpn}.  All edges involving the vertex~$v$ corresponding to this wire have to be generated along it, and the removal of this wire leads to a reduced graphic parity network for the subgraph~$G - v$.  Thus, there must be precisely~$d + 1$ operations targeting this wire, where~$d$ is the number of neighbors of~$v$ in~$G$.  Except for the first and the last, each operation on this wire makes a new edge term by cancellation.
This motivates us to define \emph{perfect cancellation orderings} and \emph{perfect cancellation graphs}.  The formal definition is technical and hence deferred to \Cref{sec:perfect-cancellation}.
All chordal graphs are perfect cancellation graphs: all perfect elimination orderings~\cite{roseTriangulatedGraphsElimination1970} are perfect cancellation orderings, but not the other way.
As we will see, a perfect cancellation ordering can guide us in synthesizing a perfect graphic parity network for~$G$ in linear time.

\begin{theorem}
  \label{thm:chordal-ideal}
  All perfect cancellation graphs are \ptt.  Given a perfect cancellation graph~$G$ and a perfect cancellation ordering of~$G$, we can synthesize a perfect graphic parity network for~$G$ in linear time.
\end{theorem}

We conjecture that the two graph classes are equivalent: a graph is \ptt{} if and only if it is a perfect cancellation graph.  %
We leave it to the reader to verify that a simple cycle on four vertices is not a perfect cancellation graph, while it becomes one after adding a universal vertex, i.e., a vertex that is adjacent to all the vertices on the cycle.  Thus, the class of perfect cancellation graphs is not \emph{hereditary}, i.e., closed under taking induced subgraphs.  The same holds for \ptt{} graphs.
This suggests that both classes are not easy to \replaced{handle}{deal} algorithmically.  Indeed, finding a perfect cancellation ordering is computationally hard.

\begin{theorem}\label{thm:hardness-perfect-cancellation}
  It is \NP-complete to decide whether a graph is a perfect cancellation graph.
\end{theorem}

Finally, we present a fixed-parameter tractable algorithm for recognizing perfect cancellation graphs, using the the treewidth of the input graph as the parameter.
Similar to most algorithms using a tree decomposition, we use dynamic programming bottom-up.  The main challenge is that a subgraph may need vertices from without to make a good ordering.  As said, the class of perfect cancellation graphs is not hereditary.

\begin{theorem}
  \label{thm:recognition-fpt}
  The recognition of perfect cancellation graphs is fixed-parameter tractable parameterized by the treewidth of the input graph.
\end{theorem}

\paragraph*{Other related work.}

Two other important measures for quantum circuits or reversible circuits are the depth and the number of ancillae.
The \emph{depth} of a quantum circuit is the count of time steps needed to execute all the gates in the circuit in parallel; e.g., the first two gates in \Cref{fig:trivial-gpn} contribute  one to the depth.
An \emph{ancilla bit} is a qubit whose input is the particular state $\ket{0}$ and can be utilized as auxiliary space throughout the computation but must be recovered to $\ket{0}$ at the end of computation.
The circuit depth characterizes the running time of a quantum circuit, while the number of ancillae characterizes the extra space required by a quantum circuit.
Any $n$-qubit \textsc{cnot} circuit can be represented by an invertible matrix $\mathbb{M} \in \mathbb{F}_2^{n \times n}$, and the synthesis of \textsc{cnot} circuit is equivalent to transforming $\mathbb{M}$ to identity by Gaussian elimination; see more details in the appendix.
The main trick in minimizing circuit depth is to employ more ancillary qubits to eliminate multi-columns rather than one-column simultaneously~\cite{10.5555/3381089.3381102,Maslov2022,Goubault2024}.  Interestingly, this ultimately reduced to the parallel Gaussian elimination, which is inherently related to chordal graphs, also known as perfect elimination graphs.
It is easy to see that any parity network for a connected graph has depth $\Omega(\log n)$.  In fact, there is a trivial method synthesizing parity network for any graph in depth $O(\log n)$, as long as enough ancillae are given.
This line of work is orthogonal to ours because it usually increases the size.

\section{Graphic parity networks and \ptt{} graphs}

All graphs discussed in this paper are finite and simple. The vertex set and edge set of graph~$G$ are denoted by, respectively,~$V(G)$ and~$E(G)$.  Throughout the paper we use~$n = |V(G)|$ and~$m = |E(G)|$. %
An~$n$-qubit circuit over \textsc{cnot} gates is a \emph{graphic parity network} for a graph~$G$ if for every edge~$(u,v)$ of~$G$, 
the term $\{u,v\}$ appears in the annotated circuit and the final state of the wires is the same as the original state.
\deleted{The number of gates in a graphic parity network is called its \emph{size}.
  A graphic parity network can be alternatively viewed as a sequence of operation.  We denote the operation that adds the $i$th wire to the $j$th by $\rowadd{i,j}$, of which wire~$j$ is the \textit{target}.}
A term is \emph{singleton} or \emph{binary} if its cardinality is one or two, respectively.
For our purpose, this definition suffices, and we refer to \Cref{app:quantum-intro} for a definition requiring more quantum background.

We start with proving lower bounds announced in Theorems~\ref{thm:opt-pn-size-ctrb} and~\ref{thm:hel_nsqrt}.  These proofs have two implications.  First, we define a novel graph class that contains all chordal graphs, and show that all the graphs in this class admit the smallest possible graphic parity networks.  Second, we propose a randomized algorithm for synthesizing graphic parity networks for general graphs.

\subsection{Lower bounds}\label{sec:lower-bound}

\clu{All reverse is replaced with inverse. In quantum computing, simply read the circuit from right to left without relabeling is usually referred to as inverse.}
A parity network~$C$ can also be read from right to left, which defines another quantum circuit, called the \emph{inverse} of~$C$\deleted{, where each wire is labeled by the final output of this wire in~$C$.  As shown in Cref{fig:gpn-reversal} in the appendix, the final output of a wire may or may not be the same as the label of this wire}.
Consider a wire with~$\ell$ operations in~$C$, which defines~$\ell + 1$ terms.  It turns out that \replaced{the same wire}{the wire with the same label} in the inverse of~$C$ has the same number of terms, and they appear in exactly the inversed order as in~$C$.
This observation is formalized as the following proposition, which does not use any special properties of graphs and holds for all parity networks.  For the sake of completeness, we provide a proof in \Cref{app:lem-rev}.

\begin{proposition}[Folklore]
  \label{lem:reverse}
  The inverse of a parity network for a set is a parity network for the same set.  For each qubit wire, the wire in the inversed parity network has exactly the same terms, and they appear in the reversed order.
\end{proposition}

The following bounds the number of non-binary terms generated by a graphic parity network.

\begin{lemma}
  \label{lem:stronger}
  In any graphic parity network for a connected graph~$G$, there must be at least~$n-1$ operations whose outcomes are not binary.
\end{lemma}
\begin{proof}
  Let~$\ell$ be the size of the graphic parity network.  For~$i = 0, 1, \ldots, \ell$, we define a hypergraph~$H_{i}$ whose vertex set is~$V(G)$ and whose edge set consists of all terms generated by the first~$i$ operations, and let~$c(H_{i})$ denote the number of components of~$H_{i}$.
  By definition,
  \[
    n = c(H_{0}) \ge c(H_{1}) \ge \cdots \ge c(H_{\ell}) = 1,
  \]
  where~$c(H_{0}) = n$ because~$H_{0}$ is edgeless and~$c(H_{\ell}) = 1$ because~$G$ is connected by assumption.
  In particular,~$c(H_{i})$ is either~$c(H_{i-1})$ or~$c(H_{i-1}) - 1$, and the second case can only happen when the $i$th operation is applied to two terms that are disjoint.
  Such an operation is called a plus-operation, and by the discussion above, there are at least~$n - 1$ plus-operations.

  We take the first~$n - 1$ plus-operations of the graphic parity network, and denote them as~$C_{1}$,~$C_{2}$,~$\ldots$~$C_{n-1}$.
  For~$i = 1, \ldots, n - 1$, we select a distinct operation of which the term on the target wire \emph{before} the operation is non-binary.
  We take~$C_{i}$ if it satisfies our condition; otherwise, we take the next non-plus operation~$C'_{i}$ with the same target wire as~$C_{i}$.
  Note that it exist because the final state of this wire is singleton.
  Since the term on the target wire is binary before~$C_{i}$, and all operations between~$C_{i}$ and~$C'_{i}$ are plus-operations, the term on the target wire before~$C'_{i}$ is non-binary.
  All the~$n-1$ selected operations are distinct by the selection.
  \replaced{In the inverse of this graphic parity network, these $n-1$ non-binary terms appears \textit{after} the operations, namely, there are $n-1$ operations generating non-binary terms. And by \cref{lem:reverse}, if a term is generated in the inverse, then it must be generated in the origin, }{By proposition~\ref{lem:reverse}, the~$n-1$ non-binary terms are outcomes of operations in the inverse of this graphic parity network, } and this concludes the proof.
\end{proof}

Since there are at least~$m$ operations generating binary terms, \Cref{thm:opt-pn-size-ctrb} follows from Lemma~\ref{lem:stronger} as a corollary.
We say that a graphic parity network is \emph{perfect} if its size is precisely~$m+n-c$.
By Lemma~\ref{lem:stronger}, there is a one-to-one mapping between~$E(G)$ and the binary terms of a perfect graphic parity network.  Moreover, if~$m \gg n$, most of the terms for edges in~$G$ are generated by cancellation.
Indeed, all chordal graphs admit perfect graphic parity networks.  This can be extended to graphs close to chordal, e.g., the graph in \Cref{fig:k2p}, where~$p = n - 2$.  This graph can be turned into a chordal graph by adding a single edge, namely, $(1, 2)$, and hence it admits a graphic parity network of size~$m + n = 3p + 3$.  Note that all the induced cycles in~$K_{2, p}$ have length four.
If the girth of a graph is greater than four, the bound in \Cref{thm:opt-pn-size-ctrb} can be greatly improved.

\begin{lemma}\label{lem:upper-bound}
  Let~$G$ be a graph of girth at least five.  The minimum size of graphic parity networks for~$G$ is~$m + \Omega(m)$.
\end{lemma}
\begin{proof}
  Let us fix a graphic parity network for~$G$.
  For each edge~$e\in E(G)$, let~$f(e)$ denote the operation that generates the term corresponding to~$e$, and~$c(e)$ the immediately next operation targeting the same wire as~$f(e)$; note that~$c(e)$ exists because the final state of this wire is a singleton term.  Let~$F = \{f(e)\mid e\in E(G)\}$ and~$C = \{c(e)\mid e\in E(G)\}$.  Moreover, let~$R$ denote all the operations not in~$F$.
  Note that $|F| = |C| = m$, and the size of the network is
  \[
    |F| + |R| \ge |F| + |C\setminus F| = |F| + |C| - |F \cap C| =2 m -|F\cap C|.
  \]
  We are done if~$|F\cap C| \le m/2$.  Hence, we assume that~$|F\cap C| > m/2$.  By definition, each operation in~$F\cap C$ is~$c(e_{1}) = f(e_{2})$ for two different edges~$e_{1}$ and~$e_{2}$.
  There are two cases,
  \[
    \begin{cases}
      \ket{v, w} \boldsymbol{\cdot}\!\!-\!\!\oplus \ket{u, v}       & e_{1} = (u, v), e_{2} = (u, w),
      \\
      \ket{u, v, w, x} \boldsymbol{\cdot}\!\!-\!\!\oplus \ket{u, v} & e_{1} = (u, v), e_{2} = (w, x),
    \end{cases}
  \]
  where the three or four vertices involved in the operation are all distinct.  We use~$F_{1}$ and~$F_{2}$ to denote the sets of these two types of operations.  Note that~$F\cap C = F_{1}\cup F_{2}$.

  Case 1.  Since~$G$ does not contain any triangles,~$(v, w)$ is not an edge.  In other words,~$(v, w)$ is generated by an operation in~$R$.
  Moreover,
  if there is another operation~$c(e'_{1}) = f(e'_{2})$ in~$F_{1}$ using~$(v, w)$, then~$e'_{1} = (x, v)$ and~$e'_{2} = (x, w)$ for some vertex~$x\ne u$.  But then~$u v x w$ is a cycle of length four, violating the assumption.  Therefore, each operation in~$F_{1}$ corresponds to a distinct operation in~$R$.

  Case 2.  The term~$\ket{u, v, w, x}$ is generated by an operation in~$R$.  There are at most three operations in~$F_{2}$ using the term~$\ket{u, v, w, x}$.  Therefore, there are at least~$|F_{2}|/3$ such $4$-terms generated by~$R$.

  In summary,
  \[
    |R| \ge |F_{1}| + \frac{|F_{2}|}{3} \ge
    \frac{|F_{1}| + |F_{2}|}{3} = \frac{|F\cap C|}{3} > \frac{m}{6}.
  \]
  This concludes the proof.
\end{proof}

Theorem~\ref{thm:hel_nsqrt} follows from Lemma~\ref{lem:upper-bound} and the following lemma.  The proof of \Cref{lem:girth-5}, which is based on a classic result from extremal combinatorics, is deferred to the appendix.

\begin{lemma}\label{lem:girth-5}
  For any positive integer~$n$, there exists a graph~$G$ that has~$\Omega(n\sqrt{n})$ edges and whose girth is at least five.
\end{lemma}

A graphic parity network of size~$m+n-c$ is called \emph{perfect}, and a graph is called \emph{\ptt{}} if it admits a PERfect graphic PArity NEtwork.

\subsection{Perfect cancellation graphs}\label{sec:perfect-cancellation}

Let~$N(v)$ denote the neighborhood of~$v$, and~$N(U) = \bigcup_{v\in U} N(v) \setminus U$ for a vertex set~$U\subseteq V(G)$.
Let~$\sigma:V(G)\mapsto [n]$ be an ordering of the vertices of~$G$, where~$[n] = \{1, 2, \ldots, n\}$.
A subset~$U \subseteq V(G)$ is \emph{$\sigma$-linked} if every two consecutive vertices in~$\sigma|_{U}$, the sub-ordering of $\sigma$ induced by $U$, are adjacent in~$G$.
We use~$x <_\sigma y$ (resp.,~$x \le_\sigma y$) to denote~$\sigma(x) < \sigma(y)$ (resp.,~$\sigma(x) \le \sigma(y)$).
For each vertex~$v\in V(G)$, we denote
\[
  N^+_{\sigma}(v) = \{u\in N(v) \mid \sigma(u)>\sigma(v)\}.
\]

\begin{definition}
  An ordering~$\sigma:V(G)\mapsto[n]$ is a \emph{perfect cancellation ordering} of~$G$ if for all vertices~$v\in V(G)$ and for all components~$C$ of~$G - v$, the set~$N^+_{\sigma}(v)\cap C$ is~$\sigma$-linked.
  A graph $G$ is a \emph{perfect cancellation graph} if it has a perfect cancellation ordering.
\end{definition}

The vertex set of a chordal graph can be ordered such that, each vertex~$v$ and its neighbors that occur after $v$ in the order form a clique; such an order is called a \emph{perfect elimination ordering}~\cite{roseTriangulatedGraphsElimination1970}.
Since a clique is~$\sigma$-linked for any ordering~$\sigma$, a perfect elimination ordering is a perfect cancellation ordering \cite{roseTriangulatedGraphsElimination1970}.  In other words, all chordal graphs are perfect cancellation graphs.
It is worth noting that a perfect cancellation ordering of a chordal graph is not necessarily a perfect elimination ordering.
For example, both~$(1, 2, 3, 4)$ and~$(4, 3, 2, 1)$ are perfect cancellation orderings, but only the second is a perfect elimination ordering.
With a perfect cancellation ordering of~$G$ given, we can synthesize a perfect graphic parity network for~$G$ in linear time.  Indeed, the circuit in \Cref{fig:optimal-gpn} was generated by \Cref{alg:sythchordal}.
For the convenience of presentation, we may start with biconnected graphs, for which the condition is simplified to~$N^+_{\sigma}(v)$ being~$\sigma$-linked for all~$v$.

\SetCommentSty{textit}
\SetKwComment{trcmt}{$\vartriangleright$
}{}
\begin{algorithm}[h]
  \clu[inline]{Added detail on how to restore}
  \caption{A synthesizing algorithm for perfect cancellation graphs}
  \label{alg:sythchordal}
  \SetAlgoLined
  \For(\trcmt*[f]{$[v_1,\dots,v_n]$ is a perfect cancellation ordering of $G$}){$i \gets 2, 3, \dots, n$}{
    \For{$j \gets i-1, i-2, \dots, 1$}{
      \If(\trcmt*[f]{$j \in \mathrm{term}(j)$ and $|\mathrm{term}(j)| \leq 2$}){$(v_i, v_j) \in E(G)$}{
        \If{$\mathrm{term}(j) = \{j\}$}{
          $\rowadd{i,j}$;
        }
        \Else{
          $\{j, k\} \gets \mathrm{term}(j)$\trcmt*{$k > j$ and $\mathrm{term}(k) = \{i, k\}$}
          $\rowadd{k, j}$;
        }
      }
    }
  }
  \For{$i\gets n-1, n-2, \dots, 1$}{
    $\{i,j\}\gets \mathrm{term}(i)$\trcmt*{$j > i$ and $\mathrm{term}(j) = \{j\}$}
    $\rowadd{j,i}$\;
  }
\end{algorithm}

\begin{lemma}
  \label{lem:comp-cond}
  Let~$G$ be a biconnected graph.
  Given a perfect cancellation ordering of a graph~$G$, we can synthesize a perfect graphic parity network for~$G$ in~$O(m + n)$ time.
\end{lemma}
\begin{proof}
  Let~$\sigma$ be the perfect cancellation ordering.
  We may number the vertices such that~$\sigma(v_{i}) = i$, and use Algorithm~\ref{alg:sythchordal}.  It generates all the binary terms in the main loop (lines~1--8) before restoring the singleton terms in line~9.
  Initially, $\mathrm{term}(j) = \{j\}$ for all~$j = 1, \ldots, n$.
  The algorithm maintains the following invariants.  Before the execution of the~$i$th iteration, for all~$j = 1, \ldots, n$,
  \begin{enumerate}[label={(I\arabic*)},left=10pt]
    \item\label{inv:1} $j\in \mathrm{term}(j)$ and~$|\mathrm{term}(j)| \le 2$;
    \item\label{inv:2} $\mathrm{term}(j) = \ket{j}$ if~$j \ge i$; and
    \item\label{inv:3} if $\mathrm{term}(j) = \ket{j, k}$, then~$j < k < i$, $(v_j, v_k)\in E(G)$, and~$(v_j, v_{k'})\not\in E(G)$ for all~$k'$ with~$k < k' < i$.
  \end{enumerate}
  Now we show that the invariants are maintained.
  In the iteration from the inner loop (lines~2 to~8), only~$\mathrm{term}(j)$ is modified.
  By invariant~\ref{inv:2},~$\mathrm{term}(i) = \ket{i}$, and it remains true during the~$i$th iteration of the for the main loop because wire~$i$ is never the target.
  Moreover, for each~$j < i$, wire~$j$ is the target in and only in the $j$th iteration of the inner loop.
  If~$(v_i, v_j)\not\in E(G)$, then~$\mathrm{term}(j)$ is not changed, and all the invariants remain true.
  Hence, assume~$(v_i, v_j)\in E(G)$, and we argue that~$\mathrm{term}(j) =\ket{i, j}$ after this iteration, and then all invariants remain satisfied afterward.
  It is straightforward when~$\mathrm{term}(j) =\ket{j}$ (line~5), and we focus on the else branch (line~6).
  Line~7 is correct by~\ref{inv:1} and the fact that the condition in line~4 is not satisfied.
  By invariant~\ref{inv:3}, $k > j$, and~$(v_{k'}, v_j)\not\in E(G)$ for all~$k'$ with~$k < k' < i$.
  Since~$\sigma$ is a perfect cancellation ordering,~$k$ and~$i$ are adjacent.
  By invariant~\ref{inv:3},~$\mathrm{term}(k) =\ket{k, i}$, and thus line~8 sets~$\mathrm{term}(j)$ to~$\ket{k, i}$.
  Thus, the algorithm correctly produces a graphic parity network for~$G$.

  We now verify that the synthesized circuit is perfect.  For each edge, the algorithm introduces precisely one gate.  Moreover, since~$G$ is connected, by invariant~\ref{inv:3},~$|\mathrm{term}(j)| = 2$ for all~$j = 1, \ldots, n-1$ when the algorithm reaches line~9.
  By invariants~\ref{inv:1} and~\ref{inv:3}, we can restore every wire to be a singleton term by one gate.  Thus, the total size is precisely~$m + n - 1$.

  Let us briefly explain the implementation.
  We assume the graph is stored as adjacency lists.  It is pedestrian to reconstruct the lists such that each list is sorted.  Thus, the number of iterations of the loop of line~3 can be the number of neighbors of~$v$.  The total time is thus~$O(m + n)$
\end{proof}

\clu{A possibly simpler way of working with non-biconnected graphs}
If G is not biconnected, we %
synthesize a graphic parity network for $G$ by independently handling its biconnected components. The detailed description as well as the formal proof of \cref{thm:chordal-ideal} is deferred to \cref{app:prof-chordal-ideal}.

\subsection{A randomized synthesizing algorithm}
\label{sec:randomized-algo}

Let~$G$ be an arbitrary graph.  We present a random algorithm to synthesize a graphic parity network for~$G$.
We process the vertices in a random order, and for each vertex~$i$, we generate all the edges between this vertex and latter vertices in this order, before resetting this wire to~$\{i\}$.
Similar to \Cref{alg:sythchordal}, we never introduce a term with more than two elements, and the item~$i$ never leaves wire~$i$.
It has three phases.

The first phase only applies when the term on wire~$i$ is binary, i.e., it has an earlier neighbor in the order.  Let~$j$ be the other number in this term.  For all unprocessed neighbors~$v_k$ of~$v_i$, if the term on wire~$k$ is~$\{j, k\}$, we can generate~$\{i, k\}$ by adding~$\{i, j\}$ to it.  The three vertices form a triangle.

In the second phase, we deal with neighbors~$v_j$ of~$v_i$ such that the term on wire~$j$ is binary.  We group them according to the other item in their terms, and process each group by cancellation.  Here we attempt to add a non-edge term~$\ket{i, j}$ to facilitate dealing with edges between~$v_i$ and common neighbors of~$v_i$ and~$v_j$; see the discussion about~$K_{2, p}$ above for motivation.

Finally, we deal with other neighbors of~$v_i$ individually.
We summarize it as Algorithm~\ref{alg:sythrandom}.

\begin{algorithm}[h]
  \caption{A randomized algorithm for graphic parity network synthesis.}
  \label{alg:sythrandom}
  \SetAlgoLined
  renumber the vertices such that~$\pi(v_{i}) = i$\trcmt*[r]{$\pi$ is a random permutation of~$[n]$.}
  \For(\trcmt*[f]{The size of each term is at most two.}){$i \gets 1, 2, \dots, n$}{
    \If{$\mathrm{term}(i)$ is binary}{
      $\{i, j\} \gets \mathrm{term}(i)$\trcmt*[r]{$j < i$.}
      \For(\trcmt*[f]{$k > i$.}){$v_k \in N(v_i)$ \textbf{such that} $\mathrm{term}(k) = \{j, k\}$}{
        \Call{\textsc{cnot}}{i, k}\;
        remove edge $(v_i, v_k)$\;
      }
      \Call{\textsc{cnot}}{j, i}\;
    }
    \For{$j \gets 1, 2, \dots, i-1$}{
      $K \gets \{v_k \in N(v_i) \mid \mathrm{term}(k) = \{j, k\}\}$\;
      \If(\trcmt*[f]{$v_{i} v_{j}\not\in E(G)$ or is already generated.}){$K \neq \emptyset$}{
        \Call{\textsc{cnot}}{j, i}\;
        \For(\trcmt*[f]{$k > i$.}){$v_k \in K$}{
          \Call{\textsc{cnot}}{i, k}\;
          remove edge $(v_i, v_k)$\;
        }
        \Call{\textsc{cnot}}{j, i}\trcmt*[f]{Clean up.}
      }
    }
    \For(\trcmt*[f]{$k > i$.}){$v_k \in N(v_i)$}{
      \Call{\textsc{cnot}}{i, k}\;
      remove edge $(v_i, v_k)$\;
    }
  }
\end{algorithm}

\begin{lemma}\label{lem:random-correctness}
  Algorithm~\ref{alg:sythrandom} synthesizes a graphic parity network for~$G$ in polynomial time.
\end{lemma}
\begin{proof}%
  The algorithm starts with renumbering the vertices such that~$\pi(v_{i}) = i$.
  Initially,~$\mathrm{term}(i) = \{i\}$ for all~$i = 1, \ldots, n$.
  In the~$i$th iteration of the main loop (lines 2--19), the algorithm generates all the terms for edges~$(v_{i}, v_{k})$ with~$k > i$, before reseting~$\mathrm{term}(i) = \{i\}$.
  If~$\mathrm{term}(i)$ is binary at the beginning, line~8 resets it.
  The only operations targeting wire~$i$ are lines 12 and 16.  If line 12 changes wire~$i$, line 16 duly resets it.
  It remains to verify that all edges are generated.
  Let~$v_{k}$ be a neighbor of~$v_{i}$ with~$k > i$.
  If~$\mathrm{term}(k)$ is not binary before the~$i$th main loop, it is added by either line 6 or line 15, depending upon whether~$\mathrm{term}(i)$ and~$\mathrm{term}(k)$ have a item.
  Thus, the algorithm correctly produces a graphic parity network for~$G$.  The algorithm clearly runs in polynomial time.
\end{proof} 
In \Cref{alg:sythrandom}, except for lines 12 and 16, each other operation either generates a new item or clears up a wire.  Therefore, to bound the \deleted{number of }size of the synthesized circuit, it suffices to bound the number of them being executed.  Since they are always executed in pair, it suffices to count line 12.
For a set~$X$, we use~${\mathfrak {S}}_{X}$ to denote the set of all permutations of~$X$, and we use~${\mathfrak {S}}_{n}$ as an shorthand for~${\mathfrak {S}}_{[n]}$.
The degree of vertex~$v$ is~$d(v)$.

\begin{lemma}\label{lem:analysis}
  The expected number of line~12 being executed is upper bounded by
  \[
    4n+ \min_{1\le t\le n}\left\{2\sum_{i< t}d_i+\frac{(n-t)n\log n}{d_{t}}\right\},
  \]
  where $d_1,d_2,\ldots,d_n$ are the degrees of the vertices sorted in ascending order.
\end{lemma}
\begin{proof}
  Line~12 is only executed when~$K\ne \emptyset$ and~$i\not\in \mathrm{term}(j)$ at the beginning of the~$i$th iteration.  With permutation~$\pi$, the number of line~12 being executed is the number of pairs~$(u, v)$ such that there exists an extra vertex~$w$ satisfying
  \begin{enumerate}[label = (\roman*),left=10pt]
    \item $u w, v w\in E(G)$;
    \item $u$ is not the last neighbor of~$v$ in~$\pi$;
    \item $\pi(u) < \pi(v) < \pi(w)$; and
    \item $w x\not\in E(G)$ for all vertices~$x$ with $\pi(u) < \pi(x) < \pi(v)$.
  \end{enumerate}
  Therefore, it cannot be more than the number of pairs~$(u, v)$ such that there exists an extra vertex~$w$ merely satisfying (iv) and
  \begin{enumerate}[label = (\roman*),left=10pt]
    \setcounter{enumi}{3}
    \item $u w, v w\in E(G)$ and $\pi(u) < \pi(v)$.
  \end{enumerate}
  Now we obtain an upper bound on the number of such pairs.

  For any permutation~$\pi \in {\mathfrak {S}}_{n}$,
  let~$P(\pi)$ denote all such pairs when the algorithm is executed using permutation $\pi$, and let~$P(\pi, j)$ denote the subset of~$P(\pi)$ in which the first vertex is fixed by~$\pi(u) = j$.
  Then the expected number of such pairs is
  \begin{align}
    \notag
    \mathop{\mathbb{E}}_{\pi \in {\mathfrak {S}}_{n}} |P(\pi)| = & \mathop{\mathbb{E}}_{\pi \in {\mathfrak {S}}_{n}} \sum_{j = 1}^{n} |P(\pi, j)|
    \\      \notag
    =                                                            & \sum_{j = 1}^{n} \mathop{\mathbb{E}}_{\pi \in {\mathfrak {S}}_{n}}  |P(\pi, j)|
    \\      \notag
    =                                                            & \sum_{j=1}^n
    \mathop{\mathbb{E}}_{X\in {[n]\choose j}}
    \mathop{\mathbb{E}}_{\substack{{\pi \in {\mathfrak {S}}_{X}}                                                                                           \\  \forall x\in X. \pi(x) > n - j}}
    |P(\pi, n - j + 1)|
    \\
    \le                                                          & \sum_{j=1}^n \mathop{\mathbb{E}}_{\pi \in {\mathfrak {S}}_{n}} |P(\pi, 1)|.\label{eq:4}
  \end{align}
  Therefore, it reduces to bounding the expected size of~$P(\pi, 1)$.
  Consider any fixed~$t\in [n]$.
  If~$d(\pi^{-1}(1)) \le d_{t}$, then
  \begin{equation}
    \label{eq:5}
    |P(\pi, 1)| \le d_{t}.
  \end{equation}
  We now consider the nontrivial case, where~$d(\pi^{-1}(1)) > d_{t}$.
  Note that~$|P(\pi, 1)|$ is the number of vertices~$v$ such that there exists another vertex~$w$ whose first two neighbors in~$\pi$ are~$\pi^{-1}(1)$ and~$v$; we say that it is \emph{witnessed by~$w$}.
  For each vertex~$w\in N(\pi^{-1}(1))$, let~$X_{w}$ be the index of the second neighbor of~$x$ in~$\pi$.  Then
  \[
    \mathbb{P}\left( X_{w}=i\right) = \left(1-\frac {d(w)-1}{n-2}\right)\cdots \left(1-\frac {d(w)-1}{n-(i-1)}\right) \frac{d(w)-1}{n-i} >                                                   \left(1- \frac {d(w)-1}{n} \right)^{i - 2} \frac{d(w)-1}{n}.
  \]
  Letting~$\alpha = 1-\frac {d(w)-1}{n}$, we have
  \begin{align}
    \notag
      & \mathbb{P}\left(X_{w}\le {\frac {n \log n}{d(w)-1}} \right)
    \\
    \notag
    = & \mathbb{P}\left(X_{w}=2  \right) + \cdots +
    \mathbb{P}\left(X_{w}=\left\lfloor {\frac {n \log n}{d(w)-1}} \right\rfloor \right)
    \notag
    \\
    > & \frac{d(w)-1}{n} + \cdots + \alpha^{\left\lfloor {\frac {n \log n}{d(w)-1}} \right\rfloor - 2} \left(\frac{d(w)-1}{n}\right)
    \notag
    \\
    = & \left(\frac{1 - \alpha^{\left\lfloor {\frac {n \log n}{d(w)-1}} \right\rfloor - 1}}{1 - \alpha}\right) \left(\frac{d(w)-1}{n}\right)
    \notag
    \\
    = & {1 - \alpha^{\left\lfloor {\frac {n \log n}{d(w)-1}} \right\rfloor - 1}}.
    \notag
    \label{eq:2}
  \end{align}
  If~$d(w) > d_t$, then~$\frac {n\log n}{d(w)-1} \le \frac {n\log n}{d_t}$, and
  \begin{align}
    \mathbb{P}\left(X_{w}>\frac {n\log n}{d_t}\right)< \alpha^{\left\lfloor {\frac {n \log n}{d(w)-1}} \right\rfloor - 1} =  \left(1-\frac {d(w)-1}{n}\right)^{\left\lfloor {\frac {n \log n}{d(w)-1}} \right\rfloor - 1} < \frac{4}{n}.
  \end{align}
  Since the number of such vertices~$w$ is less than $n-1$ , then all such vertices witness at most~$\frac {n\log n}{d_t} + 4$ vertices.
  On the other hand, each vertex~$w$ with~$d(w)\le d_t$ witnesses at most one vertex.
  In summary, for each~$u\in V(G)$,
  \begin{equation}
    \label{eq:3}
    \mathop{\mathbb{E}}_{\substack{\pi \in {\mathfrak {S}_{n}}\\ \pi(u) = 1}} |P(\pi, 1)| <
    \frac {n\log n}{d_t} + 4 +
    \sum_{\substack{{w\in N(u)}\\ {d(w)\le d_t}}}1.
  \end{equation}

  Combining~\eqref{eq:4},~\eqref{eq:5}, and~\eqref{eq:3}, we conclude that the expected number of line~12 being executed is less than
  \begin{equation}
    \label{eq:6}
    \sum_{d_i\le t} d_i + \sum_{\substack{u \in V(G)\\ d(u) > d_{t}}} \left(\frac {n\log n}{d_t} + 4 +
    \sum_{\substack{{w\in N(u)}\\ {d(w)\le d_t}}}1\right)
    = 2\sum_{d_i\le t} d_i + \sum_{d_i>t} \frac {n\log n}{d_t} + 4n.
  \end{equation}
  The statement follows because~\eqref{eq:6} holds for all~$t\in [n]$.
\end{proof}

{\Cref{cor:simpl} is thus a direct consequence of \Cref{lem:random-correctness} and \Cref{lem:analysis}, whose proof is defered to \cref{app:proof-cor:simpl}.}

\section{Recognition of perfect cancellation graphs}

For any graph class, the first algorithmic question is its \emph{recognition}: to decide whether a given graph is in this class.
In this section, we investigate the complexity and algorithms of recognizing perfect cancellation graphs.

\subsection{The NP-completeness of recognition}

\replaced{We now prove \cref{thm:hardness-perfect-cancellation}, showing }{We show} that the recognition of perfect cancellation graphs is \NP-compete, by a reduction from the following problem, which is shown to be \NP-complete by Opatrny~\cite{opatrnyTotalOrderingProblem1979}.
\begin{definition}[Betweeness]
  Given a finite set~$S$ and a set of ordered triples~$T\subseteq S\times S\times S$, the betweenness problem asks to determine whether there exists a total ordering~$\pi$ of S such that for every triple~$(x,y,z)$ in~$T$, either~$x <_\pi y <_\pi z$ or~$z <_\pi y <_\pi x$.
\end{definition}

The key observation of our reduction is to use a set of false twins to force the order on a triple of vertices with two edges among them.  A set of vertices with the same neighborhood is called a \emph{false twins}.  Note that by definition, there cannot be any edge among false twins.

\begin{lemma}
  \label{lem:force-edge-orientation}
  Let~$G$ be a perfect cancellation graph and~$I$ a set of false twins of~$G$.  If~$|N(I)|\le |I| - 1$, then~$N(I)$ is~$\sigma$-linked in any perfect cancellation ordering~$\sigma$ of~$G$,
\end{lemma}
\begin{proof}
  The statement holds vacuously if~$I$ comprises a single vertex.  Hence, we assume that~$|I| \ge 2$, and hence no vertex in~$I$ is a cut vertex.
  Let~$v$ be the first vertex in~$\sigma$ from~$I\cup N(I)$.
  It suffices to show that~$v\in I$: note that~$N^{+}_{\sigma}(v) = N(v) = N(I)$.
  Suppose for contradiction that~$v\not\in I$, i.e., $v\in N(I)$.  We may number the vertices in~$N^{+}_{\sigma}(v)$ as~$u_{1}, \ldots, u_{\ell}$, where~$\ell = |N^{+}_{\sigma}(v)|$, such that~$u_{i} <_\sigma u_{i+1}$ for all~$i = 1, \ldots, \ell - 1$.
  Since there is no edge among vertices in~$I$, between any two of them there is another vertex, which has to be from~$N(I)$.
  This is nevertheless impossible because there are at most~$|N(I)\setminus \{v\}| \le |I| - 2$ such vertices.
\end{proof}

In Figure~\ref{fig:gadget}, for example, in any perfect cancellation ordering~$\sigma$, either~$v_1 <_\sigma v_2 <_\sigma v_3$ or~$v_3 <_\sigma v_2 <_\sigma v_1$.

\begin{figure}[ht]
  \centering \small
  \begin{tikzpicture}
    \def\n{3}
    \def\m{4}

    \foreach \i in {1,...,\m}{
        \node[empty vertex, "$z_\i$" above] (t\i) at (\i - 1, 1) {};
      }
    \foreach \i in {1,...,\n}{
        \node[filled vertex, "$v_\i$" below] (a\i) at ({1.5*(\i - 1)}, 0) {};
        \foreach \j in {-1, 0, 1}
        \draw (a\i) -- ++(\j/9, -.15);
      }
    \foreach \i in {1,...,\n}{
        \foreach \j in {1,...,\m}{
            \draw (t\j) -- (a\i);
          }
      }
    \foreach \i in {1,...,\inteval{\n-1}}{
        \draw (a\i) -- (a\inteval{\i+1});
      }
  \end{tikzpicture}.11 w
  \caption{Illustration for \Cref{lem:force-edge-orientation}, where~$I = \{z_{1}, \ldots, z_{4}\}$ and~$v_{1} v_{2} v_{3}$ is an induced path.  Note that $v_{1}, v_{2}, v_{3}$ might have other neighbors.}
  \label{fig:gadget}
\end{figure}
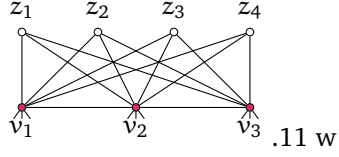

We are now ready to describe the reduction from an instance~$(S, T)$ of the betweenness problem to the recognition of perfect cancellation graphs.
Let~$p = |S|$ and~$q = |T|$.
We construct a graph~$G$ on~$2 p + 14 q$ vertices as follows.
First, for each element~$x$ in~$S$, introduce a vertex; abusing notation, we use~$x$ to denote both the element and the corresponding vertex, and use~$S$ to denote this set of vertices.
Second, we introduce a vertex set
\[
  C = \{v_{1}^{i}, v_{3}^{i} \mid 1\le i \le q\}.
\]
We add an edge between each pair of vertices in~$C$ unless they have the same superscript (note that the complement of the subgraph induced by~$C$ is an induced matching), and we add all the~$2 p q$ edges between~$S$ and~$C$.
Third, we add a set~$U$ of~$p$ vertices, and make them universal in the subgraph induced by~$S\cup C\cup U$.

Finally, we add the sets of false twins to enforce the desired order.
For convenience, we use~$({v^{i}_{0}}, {v^{i}_{2}}, {v^{i}_{4}})$ to denote the three vertices in~$S$ corresponding to the three elements in the~$i$th triple in~$T$ in order.
For each~$i = 1, \ldots, q$, we introduce 12 vertices~$z_{1}^{i}, \ldots, z_{12}^{i}$.  For each~$j = 0, 1, 2$, let
\[
  Z^{i}_{j} = \{z_{4 j + 1}^{i}, \ldots, z_{4 j + 4}^{i}\},
\]
and add all the 12 edges between~$Z^{i}_{j}$ and~$\{v_{j}^{i}, v_{j+1}^{i}, v_{j+2}^{i}\}$.
See Figure~\ref{fig:reduction} for an illustration.

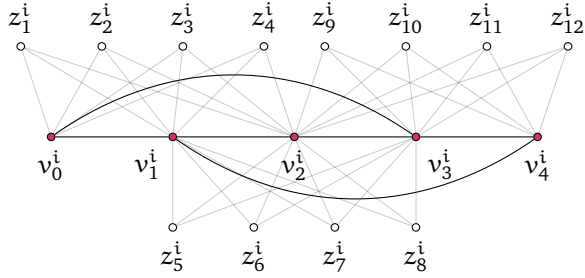
\begin{figure}[ht]
  \centering \small
  \begin{tikzpicture}[scale=.8]
    \def\n{5}
    \def\m{7}
    \def\k{4}
    \node[filled vertex, "${v^{i}_{0}}$" below] (a1) at (2,0) {};
    \node[filled vertex, "${v^{i}_{1}}$" below left] (a2) at (4,0) {};
    \node[filled vertex, "${v^{i}_{2}}$" below] (a3) at (6,0) {};
    \node[filled vertex, "${v^{i}_{3}}$" below right] (a4) at (8,0) {};
    \node[filled vertex, "${v^{i}_{4}}$" below] (a5) at (10,0) {};

    \path (a1) edge[bend left=35] (a4);
    \path (a2) edge[bend right=35] (a5);

    \def\percents{{0,0.0, 0.16666666666666666, 0.3333333333333333, 0.5, 0.6666666666666666, 0.8333333333333334, 1.0}};
    \foreach \i in {1,...,\k}{
        \coordinate(x\i) at ($(a1)!\percents[\i]!(a\n)$);
        \node[empty vertex, "$z^{i}_{\i}$" above] (t\i) at ($(x\i)+(-.5, 1.5)$) {};
      }

    \foreach \i in {\inteval{\k+1},...,8}{
        \coordinate(x\i) at ($(a1)!\percents[\inteval{\i-1}]!(a\n)$);
        \node[empty vertex, "$z^{i}_{\inteval{\i+4}}$" above] (t\i) at ($(x\i)+(.5, 1.5)$) {};
      }

    \def\percents{{0,0,0.3333333,0.66666666666666666,1}};
    \foreach \i in {1,...,\k}{
        \node[empty vertex, "$z^{i}_{\inteval{\i+4}}$" below] (p\i) at ($(a2)!\percents[\i]!(a4)+(0, -1.5)$) {};
      }

    \foreach \j in {1,...,\k}{
        \foreach \i in {1,2,3}{
            \draw[opacity=0.2] (t\j) -- (a\i);
          }
        \foreach \i in {2,3,4}{
            \draw[opacity=0.2] (p\j) -- (a\i);
          }
        \foreach \i in {3,4,5}{
            \draw[opacity=0.2] (t\inteval{\j+4}) -- (a\i);
          }
      }
    \foreach \i in {1,...,\inteval{\n-1}}{
        \draw (a\i) -- (a\inteval{\i+1});
      }
  \end{tikzpicture}
  \caption{The construction for the proof of \Cref{thm:hardness-perfect-cancellation}.}
  \label{fig:reduction}
\end{figure}

\begin{proof}[Proof of \Cref{thm:hardness-perfect-cancellation}]
  It is easy to check whether an ordering is a perfect cancellation ordering, the recognition of perfect cancellation graphs is in \NP.  For its \NP-hardness, we show that~$(S, T)$ is a yes-instance of the betweenness problem if and only if the graph~$G$ constructed above is a perfect cancellation graph.

  For sufficiency, suppose that~$G$ is a perfect cancellation graph, and let~$\sigma: V(G)\mapsto [2 p + 14 q]$ be a perfect cancellation ordering of~$G$.
  For~$i = 1, \ldots, q$, \Cref{lem:force-edge-orientation} applied to the set~$Z^{i}_{1}$ forces that
  either~$v^{i}_{1} <_\sigma v^{i}_{2} <_\sigma v^{i}_{3}$ or~$v^{i}_{3} <_\sigma v^{i}_{2} <_\sigma v^{i}_{1}$.
  We may assume without loss of generality that
  \[
    v^{i}_{1} <_\sigma v^{i}_{2} <_\sigma v^{i}_{3},
  \]
  and the other is symmetric.
  Then \Cref{lem:force-edge-orientation} applied to sets~$Z^{i}_{0}$ and~$Z^{i}_{2}$ forces that
  \[
    v^{i}_{0} <_\sigma v^{i}_{1} <_\sigma v^{i}_{2}
    \,\text{  and } v^{i}_{2} <_\sigma v^{i}_{3} <_\sigma v^{i}_{4},
  \]
  Thus, %
  \[
    v^{i}_{0} <_\sigma v^{i}_{2} <_\sigma v^{i}_{4},
  \]
  and~$\pi = \sigma|_{S}$ is a valid ordering of~$S$ that certificates that~$(S, T)$ is a yes-instance.

  For necessity, suppose that that~$(S, T)$ is an yes-instance of the betweenness problem, and let~$\pi$ be a valid ordering of~$S$.
  We may number the elements in~$S$ such that~$\pi = \langle x_{1}, \ldots, x_{p} \rangle$.
  Since reversing triples in~$T$ does not change the instance, we may assume without loss of generality that for all~$i = 1, \ldots, q$,
  \[
    v_{0}^{i} <_\pi v_{2}^{i}<_\pi v_{4}^{i}.
  \]
  For~$i = 1, \ldots, p$, let~$c(i)$ denote the number such that~$x_{c(i)} = v_{2}^{i}$; note that
  \[
    1 < c(i) < p.
  \]
  We number the vertices in~$U$ as~$\{u_{1}, \ldots, u_{p}\}$.  %

  We construct a perfect cancellation ordering~$\sigma$ of $V(G)$ as follows.
  It starts from
  \[
    z_{1}^{1}, \ldots, z_{12}^{1}, \ldots, z_{12}^{q}, u_{1}, x_{1}, \ldots, u_{p}, x_{p}.
  \]
  For all~$i = 1, \ldots, q$, we put~$v_{1}^{i}$ in between~$x_{c(i)-1}$ and~$u_{c(i)}$ and~$v_{3}^{i}$ in between~$x_{c(i)}$ and~$u_{c(i)+1}$; vertices assigned to the same range are in an arbitrary order.
  Note that
  \[
    \sigma(x_{i}) = 12 q + 2 i + 2 |\{j\mid c(j) < i\}| + |\{j\mid c(j) = i\}|.
  \]

  Now we verify that~$\sigma$ is a perfect cancellation ordering of~$G$.
  By construction,
  \[
    v^{i}_{0} \le_\sigma x_{c(i)-1}<_\sigma v^{i}_{1} <_\sigma x_{c(i)} = v^{i}_{2} <_\sigma v^{i}_{3} <_\sigma x_{c(i)+1} \le_\sigma v^{i}_{4}.
  \]
  Thus, $N^{+}_{\sigma}(z_j^i) = N(z_j^i)$ is $\sigma$-linked for all~$i = 1, \ldots, q$ and~$j = 1, \ldots, 12$.  %
  If two vertices after~$z_{12}^{q}$ in~$\sigma$ are not adjacent, they are either~$x_{j}$ and~$x_{i}$ for~$1\le j < i\le q$, or~$v_{1}^{i}$ and~$v_{3}^{i}$ for some~$i$.  Such a pair is always separated by the vertex~$u_{i}$.
  Since~$u_{i}$ is universal in the subgraph induced by~$S\cup C\cup U$, the ordering~$\sigma$ is a perfect cancellation ordering of~$G$.
\end{proof}

\subsection{Recognition is fixed-parameter tractable}

Without loss of generality, we assume that the input graph is biconnected.
Let~$k$ be the treewidth of the input graph, and let us assume we have a \textit{nice} tree decomposition $\mathcal{T}=(T,\{X_t\}_{t\in V(T)})$ of width $k$ \added{(definition in \Cref{app:fpt-recg})}. During the bottom-up dynamic programming process, we can partition the vertex into~$P$, $C$, and~$F$, which denotes the set of forgotten vertices (past), the set of vertices in the current bag (current), and the set of vertices that we have not met (future).
By definition, there is no edge between $P$ and $F$.
If~$\sigma$ is a perfect cancellation ordering, then for any vertex~$v$, we have a path on $\sigma|_{N_{\sigma}^+(v)}$.
It is broken into several sub-paths with vertices in~$P$ removed, whose ends (except the original ends of the path) are all from~$C$.

\begin{definition}[Valid order]
  \label{def:vao}
  Let~$t$ be a node of~$T$.
  A permutation~$\sigma$ of $V_t$ is valid for~$t$ if
  \begin{enumerate}[label={(V\arabic*)},left=10pt]
    \item for each $v\in V_t\setminus X_t$, the set~$N_{\sigma}^{+}(v)$ is~$\sigma$-linked; and
    \item for each $v\in X_t$, if two consecutive vertices in~$\sigma|_{N_{\sigma}^{+}(v)}$ are not adjacent, they must be in~$X_t$.
  \end{enumerate}
\end{definition}
\added{We maintain the set of all valid orders for each node~$t$ in~$T$, using the \emph{canonical representation} which allows us to keep the solution space small (namely, with size of a function of the treewidth). More details can be found in the \Cref{app:fpt-recg}.}

\section{Concluding remarks}

The main open problem is whether there are \ptt{} graphs that are \added{not} perfect cancellation graphs.

\begin{conjecture}\label{conj:main}
  A graph is \ptt{} if and only if it is a perfect cancellation graph.
\end{conjecture}

We know that there are perfect graphic parity networks that do  not correspond to any ordering in an apparent way.
Thus, to answer the conjecture, we need to better understand perfect graphic parity networks, for which there are several open questions.  In particular, whether a \ptt{} graph always admits
\begin{enumerate}[label={(C\arabic*)},left=10pt]
  \item \label{condc1} a perfect graphic parity network in which all terms are singleton and binary;
  \item \label{condc2} a perfect graphic parity network in which some wire is not a \replaced{control}{source} of any operation; and
  \item \label{condc3} a perfect graphic parity network in which no qubit leaves its original wire.
\end{enumerate}
For all the properties, we have examples of perfect graphic parity networks violating them.  But we can find alternative perfect graphic parity networks with the desired properties.
Note that \Cref{lem:stronger} does not rule out the existence of terms of cardinality three or more, though it does imply that if there is such a term, there are more wires that remain singleton throughout.\clu{Added parity networks violating the conditions.}

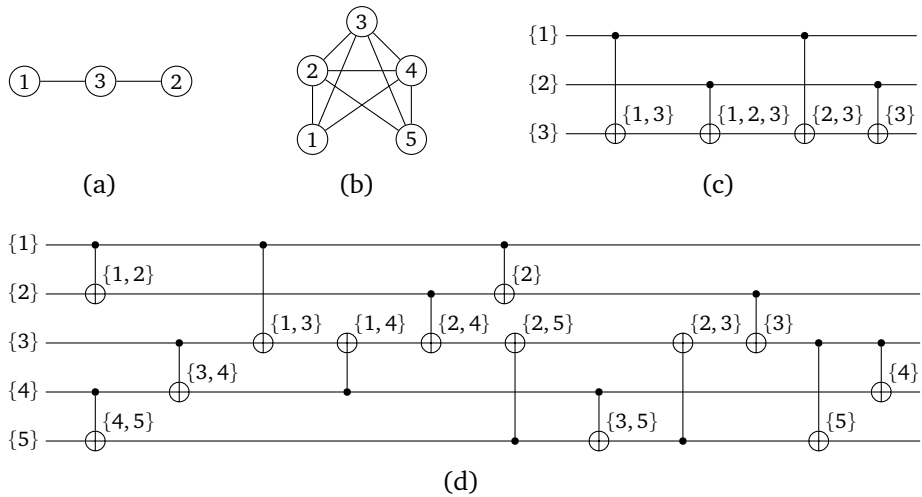
\begin{figure}[ht!]
  \centering
  \subcaptionbox{}{
    \begin{tikzpicture}[every node/.style={draw,inner sep=0pt,minimum size=4mm,circle,font=\scriptsize},scale=1]
      \node (1) at (0, 0) {$1$};
      \node (3) at (1, 0) {$3$};
      \node (2) at (2, 0) {$2$};
      \draw (1) -- (3) (3)--(2);
    \end{tikzpicture}

    \vspace{2em}
  }\hspace{1cm}
  \subcaptionbox{}{
    \begin{tikzpicture}[every node/.style={draw,inner sep=0pt,minimum size=4mm,circle,font=\scriptsize},scale=1.3]
      \node (1) at (0, 0.3) {$1$};
      \node (2) at (0, 1) {$2$};
      \node (3) at (0.5, 1.5) {$3$};
      \node (4) at (1, 1) {$4$};
      \node (5) at (1,  0.3) {$5$};
      \foreach \x/\y in {4/5, 4/3, 4/1, 4/2, 5/2, 5/3, 3/1, 3/2, 1/2}{
          \draw (\x) -- (\y);
        }
    \end{tikzpicture}
  }\hspace{1cm}
  \subcaptionbox{\label{qb:bad-examples-1}}{
    \scriptsize\begin{tikzpicture}
      \begin{yquant}[operator/separation=2mm, register/minimum height=4mm]
        qubit {$\{1\}$} x[1];
        qubit {$\{2\}$} x[+1];
        qubit {$\{3\}$} x[+1];

        \mycnot{\protect{1, 3}}{cnot x[2] | x[0]};

        \mycnot{\protect{1, 2, 3}}{cnot x[2] | x[1]};

        \mycnot{\protect{2, 3}}{cnot x[2] | x[0]};

        \mycnot{\protect{3}}{cnot x[2] | x[1]};
      \end{yquant},
    \end{tikzpicture}
    \vspace{0.5em}
  }

  \vspace{1em}

  \subcaptionbox{}{\scriptsize\begin{tikzpicture}
      \begin{yquant}[operator/separation=2mm, register/minimum height=4mm]
        qubit {$\{1\}$} x[1];
        qubit {$\{2\}$} x[+1];
        qubit {$\{3\}$} x[+1];
        qubit {$\{4\}$} x[+1];
        qubit {$\{5\}$} x[+1];
        \mycnot{\protect{4,5}}{cnot x[4] | x[3]};\mycnot{\protect{3,4}}{cnot x[3] | x[2]};\mycnot{\protect{1,2}}{cnot x[1] | x[0]};\mycnot{\protect{1,3}}{cnot x[2] | x[0]};\mycnot{\protect{1,4}}{cnot x[2] | x[3]};\mycnot{\protect{2,4}}{cnot x[2] | x[1]};\mycnot{\protect{2,5}}{cnot x[2] | x[4]};\mycnot{\protect{3,5}}{cnot x[4] | x[3]};\mycnot{\protect{2,3}}{cnot x[2] | x[4]};\mycnot{\protect{2}}{cnot x[1] | x[0]};\mycnot{\protect{3}}{cnot x[2] | x[1]};\mycnot{\protect{5}}{cnot x[4] | x[2]};\mycnot{\protect{4}}{cnot x[3] | x[2]};
      \end{yquant}
    \end{tikzpicture}
  }

  \caption{Perfect graphic parity networks that do not satisfy the conditions.\added{ {(a)} A line of length $3$. {(b)} A chordal graph. {(c)} A perfect graphic parity network for {(a)} violating \ref{condc1}. {(d)} A perfect graphic parity network for {(b)} violating \ref{condc2}, \ref{condc3}.}}
  \label{fig:bad-examples}
\end{figure}

If a graph contains a set~$U$ of~$\lfloor \frac{n}{2}\rfloor$ universal vertices, then we can make a perfect cancellation ordering by arranging the vertices in~$V(G)\setminus U$ and in~$U$ alternatively.

\begin{proposition}\label{lem:universal}
  A graph containing~$\lfloor \frac{n}{2}\rfloor$ universal vertices is a perfect cancellation graph.
\end{proposition}

An interesting question arising from \Cref{lem:universal} is: given a graph, what is the minimum number of universal vertices we need to add to make it a perfect cancellation graph?
Even more interesting is adding edges.  Since all chordal graphs are perfect cancellation graphs, this cannot be larger than the well-studied chordal completion number \cite{cai-96-hereditary-graph-modification, kaplan-99-chordal-completion}.  Can we always turn a graph into a perfect cancellation graph by adding at most~$m$ edges?  We would not bother if it needs more than~$m$: any graph has a trivial graphic parity network of size~$2 m$.

\bibliography{bib}

\newpage

\appendix

\crefalias{section}{appendix}

\section{More quantum background}\label{app:quantum-intro}

\renewcommand{\ket}[1]{|#1\rangle}

Let~$x, b\in \mathbb{F}_2^n$, then $x\cdot b = x_1b_1\oplus x_2b_2\oplus \dots\oplus x_nb_n$. %
\begin{definition}[\cite{amyControlledNOTComplexityControlledNOT2018}]
  Given $S\subseteq \mathbb{F}_2^n$, a \emph{parity network} for $S$ is an $n$-qubit \textsc{cnot} circuit that, with initial state~$\ket{b}$,
  \begin{itemize}
    \item
          for all $x\in S$ the state $\oldket{x\cdot b}$ appears on a certain qubit in the circuit; and
    \item \replaced{the final state is ~$\oldket{b_{1},b_{2},\dots,b_{n}}$.}{the final state is ~$\oldket{b_{\pi(1)},b_{\pi(2)},\dots,b_{\pi(n)}}$ for some permutation $\pi\in \mathfrak{S}_n$.}
  \end{itemize}
\end{definition}
The second requirement is from quantum algorithms like \textsc{qaa} and \textsc{qaoa}, where the original state must be restored at the end of the circuit. \added{In applications like QAOA, the output is allowed to be a permutation of the input, i.e., the final state can be ~$\oldket{b_{\pi(1)},b_{\pi(2)},\dots,b_{\pi(n)}}$ for some permutation $\pi\in \mathfrak{S}_n$. But allowing this seems not only to fail to aid in resolving the problem but also to further complicate it, so our definition restricts the output to be exactly the same as the input.}

We now briefly explain the use of graphic parity networks in \added{optimizing the quantum circuits for }\textsc{qaa} and \textsc{qaoa}.  Recall the following formulation of the maximum cut problem:
\[
  \max_{x\in \{-1,1\}^{|V(G)|}}\sum_{(u,v)\in E(G)} \frac{1-x_ux_v}{2}.
\]
In the Ising formulation~\cite{lucasIsingFormulationsMany2014}, it is equivalent to finding the ground energy (lowest eigenvalue) of the following Hamiltonian:
\[
  H_C=\sum_{(u,v)\in E}Z_uZ_v,
\]
where $Z_u$ stands for the Pauli $Z$ operator acting on the qubit $u$.
Both \textsc{qaa} as well as \textsc{qaoa} use the quantum adiabatic evolution to calculate the ground energy of $H_C$ as follows.
\begin{enumerate}
  \item Select a Hamiltonian $H_B$ such that $H_B$ has a simple ground state\deleted{, and $H_B$ and $H_C$ commute}.
  \item Initialize the quantum state to the ground state of $H_B$.
  \item Evolve the quantum state under a time-dependent Hamiltonian that gradually changes from $H_B$ to $H_C$.
  \item The final state of the quantum system is the ground state of $H_C$, whose energy can be measured to obtain the solution to the maximum cut problem.
\end{enumerate}
\renewcommand{\ket}[1]{|#1\rangle}
In the quantum circuit model, for the evolving of quantum state under a gradually changing Hamiltonian, one common approach is Trotter–Suzuki decomposition \cite{suzukiGeneralTheoryFractal1991}, where $\exp(iH_B\theta)$ and $\exp(iH_C\theta)$ should be implemented. For the implementation of $\exp(iH_C\theta)$, the parity networks come into play. For example, the naive implementation of $\exp(iH_C\theta)$ can be done by applying the following sequence of operations for all $(u,v)\in E(G)$:
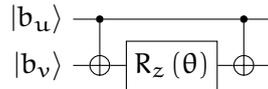
\begin{center}
  \begin{tikzpicture}
    \begin{yquant}
      qubit {$\ket{b_u}$} x[1];
      qubit {$\ket{b_v}$} x[+1];
      cnot x[1] | x[0];
      box {$R_z\left(\theta\right)$} x[1];
      cnot x[1] | x[0];
    \end{yquant},
  \end{tikzpicture}
\end{center}
where $R_z$ is the single qubit gate rotating along Pauli-$Z$ axis. Suppose the input for the circuit is $\ket{b_u,b_v}$, then the output will be $\exp(-i(b_u\oplus b_v) \theta)\ket{b_u,b_v}$, which is the desired effect of $\exp(iH_C\theta)$. So suppose we have a parity network with a small size, we can implement the operation $\exp(iH_C\theta)$ by appropriately inserting $R_z$ gates where a new binary term is generated. \added{See \cref{fig:pn-example2} for an example of how the parity network in \cref{fig:optimal-gpn} can be transformed into a circuit implementing $\exp(iH_C\theta)$.}\clu{Added how parity network transform to quantum circuits}%

\begin{figure}[ht!]
  \centering
  \begin{tikzpicture}[scale=1]
    \begin{yquant}[operator/separation=2mm]
      qubit {} x[4];

      cnot x[1] | x[0];

      box {$R_Z$} x[1];

      cnot x[2] | x[0];

      box {$R_Z$} x[2];

      cnot x[3] | x[0];

      box {$R_Z$} x[3];

      cnot x[1] | x[2];

      box {$R_Z$} x[1];

      cnot x[3] | x[2];

      box {$R_Z$} x[3];

      cnot x[2] | x[0];

      cnot x[1] | x[2];
      cnot x[3] | x[2];

    \end{yquant}
  \end{tikzpicture}
  \caption{The corresponding circuit of the parity network in \cref{fig:optimal-gpn}. It is generated by inserting a $R_z$ gate whenever and wherever a new parity term is generated.}
  \label{fig:pn-example2}

\end{figure}

A \textsc{cnot} gate can also be represented as an invertible matrix over $\mathbb{F}_2^{n \times n}$, i.e.,
\[
  {\displaystyle \textsc{cnot}(i,j)={\begin{pmatrix}1&&&&&&\\&\ddots &&&&&\\&&1&&&&\\&&&\ddots &&&\\&&1&&1&&\\&&&&&\ddots &\\&&&&&&1\end{pmatrix}}}.
\]%
Thus the transformation applied by an $n$-qubit \textsc{cnot} circuit can be written as an invertible matrix $\mathbb{M} \in \mathbb{F}_2^{n \times n}$, and the optimization of \textsc{cnot} circuit is equivalent to transforming $\mathbb{M}$ to identity using {row-addition operations, and each row-addition actually corresponds to a \textsc{cnot} gate}. %
Therefore, the optimization of \textsc{cnot} circuit is equivalent to minimizing the number of row-additions to transform an invertible matrix to identity.

\section{Proof of \cref{lem:reverse}}\label{app:lem-rev}

\begin{figure}[ht!]
  \centering\small
  \subcaptionbox{}{
    \centering
    \begin{tikzpicture}[scale=.75]
      \begin{yquant}[operator/separation=1.2mm]
        qubit {$\{1\}$} x[1];
        qubit {$\{2\}$} x[+1];
        qubit {$\{3\}$} x[+1];
        qubit {$\{4\}$} x[+1];

        \mycnot{\protect{1, 2}}{cnot x[1] | x[0]};

        \mycnot{\protect{1, 3}}{cnot x[2] | x[0]};

        \mycnot{\protect{1, 4}}{cnot x[3] | x[0]};

        \mycnot{\protect{2, 3}}{cnot x[1] | x[2]};

        \mycnot{\protect{3, 4}}{cnot x[3] | x[2]};

        \myccnot{Red4}{3}{cnot x[2] | x[0]};

        \myccnot{Red4}{2}{cnot x[1] | x[2]};
        \myccnot{Red4}{4}{cnot x[3] | x[2]};

        output {$\{1\}$} x[0];
        output {$\{2\}$} x[1];
        output {$\{3\}$} x[2];
        output {$\{4\}$} x[3];
      \end{yquant}
    \end{tikzpicture}
  }
  \hspace{0.5cm}
  \subcaptionbox{}{
    \centering
    \begin{tikzpicture}[scale=.75]
      \begin{yquant}[operator/separation=1.2mm]
        qubit {$\{1\}$} x[1];
        qubit {$\{2\}$} x[+1];
        qubit {$\{3\}$} x[+1];
        qubit {$\{4\}$} x[+1];

        \mycnot{\protect{3,4}}{cnot x[3] | x[2]};
        \mycnot{\protect{2,3}}{cnot x[1] | x[2]};
        \mycnot{\protect{1,3}}{cnot x[2] | x[0]};
        \mycnot{\protect{1, 4}}{cnot x[3] | x[2]};
        \mycnot{\protect{1, 2}}{cnot x[1] | x[2]};
        \myccnot{Red4}{\protect{4}}{cnot x[3] | x[0]};
        \myccnot{Red4}{\protect{3}}{cnot x[2] | x[0]};
        \myccnot{Red4}{\protect{2}}{cnot x[1] | x[0]};

        output {$\{1\}$} x[0];
        output {$\{2\}$} x[1];
        output {$\{3\}$} x[2];
        output {$\{4\}$} x[3];
      \end{yquant}
    \end{tikzpicture}
  }
  \caption{(a) The parity network \cref{fig:pn-example}b and (b) its inverse.}
  \label{fig:gpn-reversal}
\end{figure}

We prove by induction on the number of gates executed in the circuit. To be precise, given a parity network $C$ of $\ell$ operations and its inverse $C^{\dagger}$, for all $k=0,1,2,\dots,\ell$, we prove that if we execute the first $k$ operations in $C^{\dagger}$ and the first $\ell-k$ operations in $C$, the terms on wire $i$ in $C$ will be the same as the term on wire $i$ in $C^{\dagger}$.

\textit{Base case}: When $k=0$, the statement is trivially true, since the $C$ has finished with the term on wire $i$ in $C$ being $\{i\}$, and $C^{\dagger}$ has not started yet with the term on wire $i$ being $\{i\}$.

\textit{Inductive step}: Suppose the statement is true for $k=0,1,\dots,k_0-1$. %
We suppose the $(\ell-k_0)$-th operation in $C$ is $\rowadd{i,j}$ and thus makes the $k_0$-th operation in $C^{\dagger}$ being $\rowadd{i,j}$. By the induction hypothesis, the terms on wire $i$ in $C$ and $i$ in $C^{\dagger}$ are the same, and we denote them by $A$. Likewise, we use $B$ to denote the term on $j$. So after the $k_0$-th operation, the wire $j$ in $C^{\dagger}$ evaluates to $A\oplus B$, the same as the term on wire $j$ before the $(\ell-k_0+1)$-th operation in $C$. Therefore the statement is true for $k_0$ and the induction is complete.
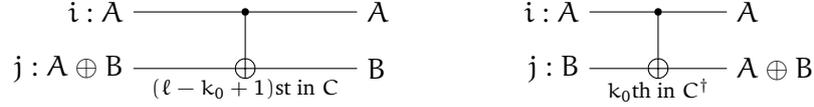
\begin{center}
  \begin{tikzpicture}
    \begin{yquant}[operator/separation=2.3mm,register/minimum height=0.5cm]
      qubit {$i:A$} x[1];
      qubit {$j:A\oplus B$} x[+1];

      output {$A$} x[0];
      output {$B$} x[1];

      ["south: $(\ell-k_0+1)$st in $C$" {font=\protect\scriptsize, inner sep=0mm}]cnot x[1] | x[0];
    \end{yquant}

    \begin{scope}[xshift=6cm]
      \begin{yquant}[operator/separation=2.3mm,register/minimum height=0.5cm]
        qubit {$i:A$} x[1];
        qubit {$j:B$} x[+1];

        output {$A$} x[0];
        output {$A\oplus B$} x[1];

        ["south: $k_0$th in $C^{\dagger}$" {font=\protect\scriptsize, inner sep=0mm}]cnot x[1] | x[0];
      \end{yquant}
    \end{scope}
  \end{tikzpicture}
\end{center}
By the induction, every term generated in $C$ has also been generated in $C^{\dagger}$ (and vice versa) and \cref{lem:reverse} is proved.

\section{Proof of \cref{lem:girth-5}}

For an integer~$k\ge 2$, a \emph{finite projective plane} of order $k$ is a plane with~$k^2+k+1$ points and~$k^2+k+1$ lines such that
\begin{enumerate}[label={(P\arabic*)},left=10pt]
  \item $k+1$ points on each line;
  \item $k+1$ lines through each point;
  \item for any two distinct points, there is a unique line passing through them; and
  \item for any two distinct lines, there is a unique point on both lines.
\end{enumerate}

\begin{theorem}[\cite{kartesziIntroductionFiniteGeometries1976}]%
  For each a prime power~$k$, there exists a finite projective plane of order $k$.
\end{theorem}
\added{
  We give the following lemma which is used in our construction, whose proof is deferred to the end of this section.
  \begin{lemma}\label{lem:getprime}
    For any \( t \ge 21 \), the largest prime number \( p \) satisfying
    \begin{equation}
      p^2 + p + 1 \le t \label{eq:primesize}
    \end{equation}
    also satisfies
    \[
      p^2 + p + 1 > \frac{t}{4}.
    \]
  \end{lemma}}
We construct a bipartite graph out of a finite projective plane as follows; see \Cref{fig:fpp-graph} for an illustration.

\begin{figure}[ht!]
  \centering\small
  \subcaptionbox{}{
    \centering
    \begin{tikzpicture}[scale=.5]
      \draw (1, 0) arc (0:360:1);
      \foreach \i in {0, 1, 2} {
          \draw ({120*\i+90}:2) -- ({120*\i-30}:2);
        }
      \foreach[count=\i from 0] \l in {a, b, c} {
          \node[empty vertex] (u\i) at ({90-120*\i}:2) {};
          \node at ({90-120*\i}:2.5) {\l};
        }
      \foreach[count=\i from 0] \l in {e, f, g} {
          \node[empty vertex] (v\i) at ({270-120*\i}:1) {};
          \node at ({270-120*\i}:1.4) {\l};
          \draw (u\i) -- (v\i);
        }
      \node[empty vertex, label={[label distance=-4pt]above right:d}] (c) at (0, 0) {};
    \end{tikzpicture}
  }
  \hspace{1cm}
  \subcaptionbox{}{
    \centering
    \begin{tikzpicture}
      \foreach[count=\i] \l in {a, ..., g}
      \node[empty vertex, "\l"] (\l) at (\i, 1) {};
      \foreach[count=\i] \list/\l in {{a,g,b}/agb, {a,f,c}/afc, {b,e,c}/bec, {a,d,e}/ade, {b,d,f}/bdf, {c,d,g}/cdg, {e,f,g}/efg} {
          \node[filled vertex,
            "\l" below] (l\i) at (\i, 0) {};
          \foreach \x in \list \draw (l\i) -- (\x);
        }
    \end{tikzpicture}
  }
  \caption{(a) The Fano plane, a finite projective plane of order 2, and (b) the bipartite graph constructed.}
  \label{fig:fpp-graph}
\end{figure}
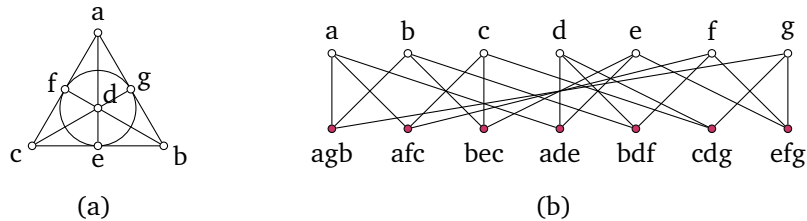

\begin{proof}
  For arbitrary~$n\ge 42$, we take~$k$ to be the largest prime power such that~$k^2+k+1 \le n/2$, and take a finite projective plane of order~$k$.
  We construct a graph~$G$ as follows.
  For each point~$p$ and each line~$\ell$ in the projective plane, we introduce a vertex, denoted by~$u_p$ and~$v_\ell$, respectively.
  We add an edge between~$u_p$ and~$v_\ell$ if and only if~$p$ is on~$\ell$.  Finally, we add~$n - 2 (k^2+k+1)$ isolated vertices to make~$n$ vertices in total.
  By (P1), (P2) and ~\cref{lem:getprime},~$m = (k^2+k+1)(k+1) = \Theta\left(n\sqrt n\right)$ \clu{The lemma is to ensure $k$ is big enough so that $m=\Omega(n\sqrt n)$.}.
  The graph cannot contain any triangle because every edge is between a point vertex and a line vertex; moreover, by (P3) and~(P4), there cannot be an induced cycle on four vertices.  Thus, the girth of~$G$ is at least five.
\end{proof}

\added{
  Finally, to prove \cref{lem:getprime}, we restate a classical result from elementary number theory:
  \begin{theorem}[Bertrand's Postulate \cite{chebyshev1852}]
    \label{thm:bertrand}
    For any integer \( n > 1 \), there exists a prime \( p \) such that \( n < p < 2n \).
  \end{theorem}
  The proof of \cref{lem:getprime} proceeds as follows:
  \begin{proof}
    Let \( x \) be the largest positive integer satisfying \( x^2 + x + 1 \le t \). When \( t \ge 21 \), \( x \) must exist and \( x \ge 4 \). This implies \( \left\lceil \frac{x}{2} \right\rceil \ge 2 \). By \cref{thm:bertrand}, there exists a prime \( q \) such that
    \[
      \left\lceil \frac{x}{2} \right\rceil < q < 2\left\lceil \frac{x}{2} \right\rceil \le x + 1,
    \]
    i.e.,
    \[
      \left\lceil \frac{x}{2} \right\rceil < q \le x.
    \]
    Since \( p \) is the largest prime satisfying \eqref{eq:primesize},
    \[
      p \ge q > \left\lceil \frac{x}{2} \right\rceil,
    \]
    which yields
    \[
      p^2 + p + 1 \ge \left\lceil \frac{x}{2} \right\rceil^2 + \left\lceil \frac{x}{2} \right\rceil + 1 > \frac{t}{4}.
    \]
  \end{proof}
}

\section{On non-biconnected perfect cancellation graphs}%
\label{app:prof-chordal-ideal}
\begin{observation}
  \label{lem:non-biconnected}
  If $G_1, G_2$ are perpane graphs with perfect parity networks $C_1, C_2$, and they share a common vertex $v$, when the vertex of $G_1, G_2$ at $v$ is perpane, with a perfect parity network constructed by simply concatenating $C_1$ and $C_2$.
\end{observation}
The proof is straightforward: $|C_1\cup C_2|=|C_1|+|C_2|=|V_1|+|E_1|-1+|V_2|+|E_2|-1=|V_1\cup V_2|+|E_1\cup E_2|-1=|V|+|E|-1$.

\Cref{lem:non-biconnected} suggests when coming accross a non-biconnected graph, we should first decompose it into biconnected components, and then synthesize a graphic parity network for each component. The final circuit is obtained by concatenating the circuits for each component. With a perfect cancellation ordering for the entire graph, we can easily obtain a perfect cancellation ordering for each component as follows:
\begin{lemma}
  \label{lem:suborder}
  Given a perfect cancellation ordering $\sigma$ of a graph $G$ and a biconnected component $G_i$ of $G$, then $\sigma':=\sigma|_{V(G_i)}$ is a perfect cancellation ordering of $G_i$.
\end{lemma}
\begin{proof}
  For a non-cut vertex $v$ in $G_i$, $N(v)\subset V(G_i)$, so $N_{\sigma'}^{+}(v)=N_{\sigma}^{+}(v)$ and remains $\sigma'$-linked. For a cut vertex $v$, let $C$ be the component of $G-v$ that contains $G_i-v$, then $(N(v)\cap C)\subset G_i-v$, so $N_{\sigma}^{+}(v)\cap C=N_{\sigma'}^{+}(v)\cap C$ remains $\sigma'$-linked. Therefore, $\sigma'$ is a perfect cancellation ordering of $G_i$.
\end{proof}

Finally, \cref{thm:chordal-ideal} is a direct consequence of \cref{lem:comp-cond}, \cref{lem:non-biconnected} and \cref{lem:suborder}, since all cut vertices and biconnected components can be obtained in linear time \cite{hopcroftjohnAlgorithm447Efficient1973}.

\section{Proof of \Cref{cor:simpl}}\label{app:proof-cor:simpl}

\begin{proof}[Proof of \Cref{cor:simpl}]
  We have seen in \Cref{lem:random-correctness} that \Cref{alg:sythrandom} always correctly synthesizes a graphic parity network.  We now analyze its expected size.
  Each operation in lines 6, 14, and 18 generates a term for a new edge.  Thus, the total number of them is~$m$.
  Line~8 is executed at most once for each vertex, and hence the total number is at most~$n$.
  Let~$s_{12}$ denote the expected number of executions of line~12.
  Therefore, the expected circuit size synthesized by this algorithm is at most~$m + 2 s_{12}+ n$, which is~$m + O(n^{1.5}\sqrt{2\log n})$ because
  \[
    \begin{aligned}
      \min_{1\le t\le n}\left\{2\sum_{i< t}d_i+\frac{(n-t)n\log n}{d_{t}}\right\} \le & \min_{1\le t\le n}\left\{2td_t+\frac{(n-t)n\log n}{d_{t}}\right\}         \\
      \le                                                                             & \min_d \left\{2nd+\frac {n^2\log n}{d}  \right\} = n^{1.5}\sqrt{2\log n}.
    \end{aligned}
  \]
  Moreover, when all but a constant number $c_1$ of vertices have degrees at least $c_2n$, we have
  \[
    \min_{1\le t\le n}\left\{2\sum_{i< t}d_i+\frac{(n-t)n\log n}{d_{t}}\right\} \le  2\sum_{d_i<c_2n} d_i + \frac {(n-c_1)n\log n}{c_2n}
    \le  2c_1\added{c_2}n+ \frac 1{c_2}n\log n.
  \]
  Thus, the size is~$m + O(n\log n)$.
\end{proof}

\section{Sketch of the parameterized algorithm}\label{app:fpt-recg}

\added{First we give the definition of nice tree decomposition. A tree decomposition $\mathcal{T}=(T,\{X_t\}_{t\in V(T)})$ is \emph{nice} if it satisfies the following properties:
\begin{itemize}
  \item \( X_r = \emptyset \) and \( X_{\ell} = \emptyset \) for every leaf \( \ell \) of \( T \). In other words, all the leaves as well as the root contain empty bags.
  \item Every non-leaf node of \( T \) is of one of the following three types:
        \begin{itemize}
          \item Introduce node: a node \( t \) with exactly one child \( t' \) such that \( X_t = X_{t'} \cup \{v\} \) for some vertex \( v \notin X_{t'} \); we say that \( v \) is \emph{introduced} at \( t \).
          \item Forget node: a node \( t \) with exactly one child \( t' \) such that \( X_t = X_{t'} \setminus \{v\} \) for some vertex \( v \in X_{t'} \); we say that \( v \) is \emph{forgotten} at \( t \).
          \item Join node: a node \( t \) with two children \( t_1, t_2 \) such that \( X_t = X_{t_1} = X_{t_2} \).
        \end{itemize}
\end{itemize}}
\added{Now we introduce the definition of canonical representation which helps us to compress the solution space.} For a node $t\in V(T)$ and a permutation~$\sigma$ of~$V_t$ and a vertex~$v\in X_t$, the \emph{canonical representation} of $\sigma|_{N_{\sigma}^{+}(v)}$ in \(\sigma\) under \(X_t\) is an ordered set of endpoints \(\{ (p_1,q_1),(p_2,q_2),\dots,(p_k,q_k)\}\) such that for all \(i\),
\begin{enumerate}
  \item \(p_i, q_i \in X_t \cup \{\dashv, \vdash\}\), only \(p_1\) is allowed to be \(\dashv\), only \(q_k\) is allowed to be \(\vdash\),
  \item segments do not intersect, i.e., \(\sigma(p_i) < \sigma(q_i) < \sigma(p_{i+1})\),
  \item the subset of elements in \(N_{\sigma}^{+}(v)\) covered by this segment, i.e., \(C_i = \{v \mid v \in N_{\sigma}^{+}(v) \wedge \sigma(p_i) \le \sigma(v) \le \sigma(q_i)\}\), is $\sigma$-linked,
  \item there are no two consecutive \(X_t\) elements in \(\sigma|_{C_i}\), and
  \item all segments together cover all $N_{\sigma}^{+}(v)$, i.e., \(\bigcup_{i=1}^{k} C_i \setminus \{\dashv, \vdash\} = N_{\sigma}^{+}(v)\).
\end{enumerate}

In our algorithm, when working from the leaves to the root, we will keep track of all canonical representations of $\sigma|_{N_{\sigma}^{+}(v)}$ for $v\in X_t$. To be precise, for each node $t\in V(T)$, we will construct a set $c_t$ that consists of all pairs $(\tau,b)$ describing a valid order $\sigma$ such that:
\begin{enumerate}[label={(V\arabic*)},left=10pt]
  \item $\tau$ is the restriction of $\sigma$ to $X_t$, and
  \item for each $v$ in the current bag, $b_v$ is the canonical representation of $\sigma|_{N_{\sigma}^{+}(v)}$ under $X_t$.
\end{enumerate}
Now we discuss how $c_t$ can be constructed in a bottom-up manner, and finally the graph is a perfect cancellation graph if and only if for the root bag $r$, $c_r\not=\emptyset$. For leaf nodes, the construction is trivial.

For a\added{n} introduce node $t$ with child $t'$ such that $X_t=X_{t'}\cup \{v\}$ for some $v$, we process all canonical representations $\tau',b'$ in $c_t'$. We enumerate the position in $\tau'$ after which $v$ is being inserted, and for a vertex $u\in N(v)$ that precedes $v$, its succeeding neighbors will also add $v$. Then we check whether $v$ is being placed between two segments in $b_u'$ since only the endpoints of the segments are allowed to have new neighbors, and if all checks pass, we can construct $b_u$ by inserting $(v,v)$ to the appropriate position in $b_u'$.
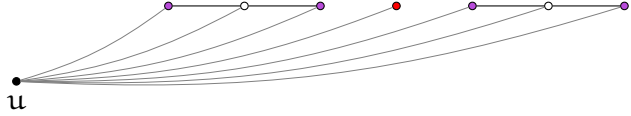
\begin{center}
  \begin{tikzpicture}[scale=1]
    \node[filled vertex,fill=black, "$u$" below] (v) at (0,-1) {};
    \readlist* \colorarray {purplecolor!80,white,purplecolor!80,red,purplecolor!80,white,purplecolor!80}
    \foreach \x in {1,...,7}{
        \node[empty vertex,fill={\colorarray[\x]}] (u\x) at (1+\x,0) {};
        \path (v) edge[bend right=10, line width=0.1,draw=gray] (u\x);
      }
    \draw (u1)--(u2)--(u3) (u5)--(u6)--(u7);
  \end{tikzpicture}
\end{center}

For a forget node $t$ with child $t'$ such that $X_t=X_{t'}\setminus \{v\}$ for some $v$, we process all canonical representations $\tau',b'$ in $c_t'$. We first check whether the succeeding neighbors of $v$ are connected, and if not, since no new neighbors of $v$ will be introduced, $N_{\sigma}^{+}(v)$ will never become $\sigma$-linked, we skip this $\tau',b'$. Then we construct $\tau,b$ by deleting $v$ from $\tau'$ and $b'$. For each $u\in N(v)$ that lies in the front of $v$, we check if $v$ can be deleted safely from its succeeding neighbors. Without loss of generality, we assume there is a segment $(v,q)$ in $b_u'$, and the predecessor of $(v,q)$ is $(p',q')$, then we check whether $q'$ is connected to $v$ (the dashed edge must exist), because if not, after $v$ is forgotten, $q'$ can never be connected to $v$ so $N_{\sigma}^{+}(u)$ can never be $\sigma$-linked. If the check passes, we should replace the segments $(p',q'), (v,q)$ with a single $(p',q)$, and $b_u$ is constructed.
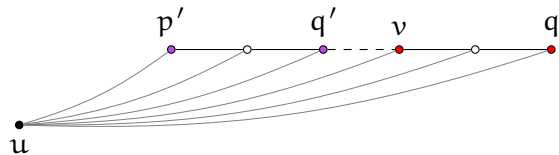
\begin{center}
  \begin{tikzpicture}[scale=1]
    \node[filled vertex,fill=black, "$u$" below] (v) at (0,-1) {};
    \def\labels{{"","$p'$","","$q'$","$v$","","$q$"}};
    \readlist* \colorarray {purplecolor!80, white,purplecolor!80,red,white,red}
    \foreach \x in {1,...,6}{
        \node[filled vertex,fill={\colorarray[\x]},label=\pgfmathparse{\labels[\x]}\pgfmathresult] (u\x) at (1+\x,0) {};
        \path (v) edge[bend right=10, line width=0.1,draw=gray] (u\x);
      }
    \draw (u1)--(u2)--(u3) (u4)--(u5)--(u6);
    \draw[dashed](u3) --(u4);
  \end{tikzpicture}
\end{center}

For a join node $t$ with children $t',t''$, then two valid orders in the two subtrees $V_{t'}$ and $V_{t''}$ can be merged only if between any pair of adjacent $X_t$ vertices in the order, either all vertices are from $V_{t'}$ or all vertices are from $V_{t''}$, otherwise the endpoint after which clash occurs (indicated red below) will not have a $\sigma$-linked succeeding neighborhood. On the other hand, once the condition holds, the full order can be constructed by arbitrarily merging the two orders, and $b_u$ can be constructed by simply merging $b'_u$ and $b''_u$ as well.
\begin{center}
  \begin{tikzpicture}
    \foreach \x in {1,...,5}{
        \node[empty vertex] (u\x) at (2*\x,0) {};
      }
    \node[empty vertex, fill=red] (u4) at (u4) {};
    \draw[-stealth] ($(u1)+(-1,0)$)--(u1);
    \draw[-stealth] (u5)-- ++(1,0);
    \foreach \i/\j in {1/2,3/4}{
        \path[Blue4] (u\i) edge[-stealth, bend left=20] (u\j);
      }
    \foreach \i/\j in {2/3,4/5}{
        \path[Blue4] (u\i) edge[-stealth, bend left=40] node[midway,empty vertex,fill=Blue4] (x\i) {} (u\j);
      }

    \foreach \i/\j in {2/3,3/4}{
        \path[Green4] (u\i) edge[-stealth, bend right=20] (u\j);
      }
    \foreach \i/\j in {1/2,4/5}{
        \path[Green4] (u\i) edge[-stealth, bend right=40] node[midway,empty vertex,fill=Green4] (y\i) {} (u\j);
      }
    \node[Blue4] at ($(u3)+(0,0.7)$) {\scriptsize order from $V_{t'}$};
    \node[Green4] at ($(u3)-(0,0.7)$) {\scriptsize order from $V_{t''}$};

    \path[red,dashed] (x4) edge[bend left=40] (y4);
  \end{tikzpicture}
\end{center}

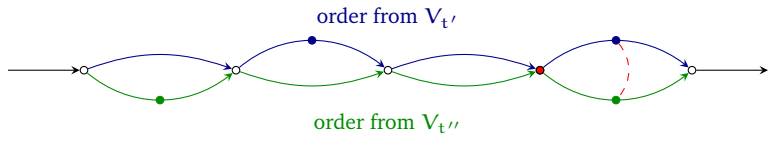
\end{document}